\documentclass{article}

\usepackage{amssymb,amsthm,amsmath}
\usepackage{mathrsfs}
\usepackage[dvipdfmx]{graphicx}
\usepackage[FIGTOPCAP,nooneline]{subfigure}

\usepackage[top=35mm,bottom=35mm]{geometry}

\usepackage{color,bm}

\makeatletter
    
    \@addtoreset{equation}{section}
\makeatother

\newcommand{\argmin}{\mathop{\rm argmin}\limits}

\theoremstyle{definition}
\newtheorem{theorem}{Theorem}[section]
\newtheorem{corollary}[theorem]{Corollary}

\newtheorem{proposition}[theorem]{Proposition}
\newtheorem{lemma}[theorem]{Lemma}
\newtheorem{remark}[theorem]{Remark}

\title{Universality of the anomalous enstrophy dissipation at the collapse of three point vortices on Euler-Poincar\'{e} models}

\author{Takeshi Gotoda \footnote{Research Institute for Electric Science, Hokkaido University, Kita 12 Nishi 7, Kita-ku, Sapporo,  Hokkaido, JAPAN  E-mail : gotoda@es.hokudai.ac.jp} 
and Takashi Sakajo \footnote{Department of mathematics, Kyoto University, Kitashirakawa Oiwake-cho, Sakyo-ku, Kyoto, JAPAN  E-mail : sakajo@math.kyoto-u.ac.jp} 
}

\date{}

\begin{document}

\maketitle

\begin{abstract}
Anomalous enstrophy dissipation of incompressible flows in the inviscid limit is a significant property characterizing two-dimensional turbulence. It indicates that the investigation of non-smooth incompressible and inviscid flows contributes to the theoretical understanding of turbulent phenomena.  In the preceding study~\cite{G.2}, a unique global weak solution to the Euler-$\alpha$ equations, which is a regularized Euler equations, for point-vortex initial data is considered, and thereby it has been shown that, as $\alpha \rightarrow 0$, the evolution of three point vortices converges to a self-similar collapsing orbit dissipating the enstrophy  in the sense of distributions at the critical time. In the present paper, to elucidate whether or not this singular orbit can be constructed independently on the regularization method, we consider a functional generalization of the Euler-$\alpha$ equations,  called the \textit{Euler-Poincar\'{e} models}, in which the incompressible velocity field is dispersively regularized by a smoothing function. We provide a sufficient condition for the existence of the singular orbit, which is applicable to many smoothing functions. As examples, we confirm that the condition is satisfied with the Gaussian regularization and  the vortex-blob regularization that are both utilized in the numerical scheme solving the Euler equations.
Consequently, the enstrophy dissipation via the collapse of three point vortices is a generic phenomenon that is not specific to the Euler-$\alpha$ equations but universal within the Euler-Poincar\'{e} models.
\end{abstract}


\section{Introduction}
In the description of two-dimensional (2D) turbulent flows at high Reynolds number, there appears a remarkable discrepancy in flow regularity between viscous flows 
in their inviscid limit and non-viscous ones. That is to say,  smooth solutions to the 2D incompressible Euler equations conserve both of the energy and the 
enstrophy, which are the $L^2$ norms of the velocity field and the scalar vorticity. On the other hand, it has been reported in \cite{Batchelor, Kraichnan, Leith} 
that the conservation of the energy and the dissipation of the enstrophy  in the inviscid limit give rise to the inertial range of the energy 
density spectrum corresponding to the backward energy cascade and the forward enstrophy cascade in 2D turbulence. The discrepancy strongly insists that turbulent flows
subject to the 2D Navier-Stokes equations converge to non-smooth flows governed by the 2D Euler equations as the Reynolds number gets infinitely large.  Hence,
the investigation of such singular flows plays a crucial role in the theoretical understanding of 2D fluid turbulence.

The first mathematical attempt to tackle this problem starts with constructing non-smooth weak solutions to the Euler equations dissipating the enstrophy. 
The global existence of a unique weak solution has been established for the initial vorticity distributions $\omega_0 \in L^1(\mathbb{R}^2) \cap L^p(\mathbb{R}^2)$ 
with $1<p \leq \infty$ \cite{Diperna, Marchioro, Yudovich}. However, it has been unfortunately shown in \cite{Eyink(a)} that weak solutions to the Euler equations for $\omega_0 \in L^1(\mathbb{R}^2)\cap L^p(\mathbb{R}^2)$ with  $p > 2$ can not dissipate the enstrophy in the sense of distributions. Therefore, it is necessary to deal with  vorticity distributions
with a weaker regularity such as  distributions in  the space of finite Radon measures $\mathcal{M}(\mathbb{R}^2)$ on $\mathbb{R}^2$.  In spite that the existence result has
been extended to the case of  $\omega_0 \in \mathcal{M}(\mathbb{R}^2)$ with a distinguished sign \cite{Delort, Majda}, it is still required that the velocity field induced by the 
vorticity distributions belongs to $L_{loc}^2(\mathbb{R}^2)$.  Consequently, if the vorticity distribution is 
represented by a $\delta$-measure, called a \textit{point vortex},  for instance, it is difficult to construct a unique global weak solution to the 2D Euler equations for this initial data,
  since its inducing velocity field is no longer the element of $L^2_{loc}(\mathbb{R}^2)$.
 To overcome this mathematical difficulty,  we  regularize incompressible velocity fields by introducing a smoothing function with a parameter $\varepsilon$.
 If we successfully construct a unique global weak solution to the equations for the regularized velocity fields, we shall obtain a non-smooth incompressible and inviscid flow dissipating the 
 enstrophy by taking the $\varepsilon \rightarrow 0$ limit of this weak solution.
 
 An example of such regularized Euler equations is the Euler-$\alpha$ equations, where $\alpha>0$ is the smoothing parameter. The Navier-Stokes-$\alpha$ and the Euler-$\alpha$
 equations are originally derived as models of 2D turbulence~\cite{Foias(a), Lunasin}. The existence of a unique global
 weak solution to the 2D Euler-$\alpha$ equations for $N$ point-vortex initial data, referred to as \textit{$\alpha$-point vortices},
 has been shown in \cite{G.2}. Then the evolution of the weak solution can be described in terms of the dynamics of those $\alpha$-point vortices.
 It was discovered in \cite{Sakajo}, and it has recently been  made mathematically rigorous in \cite{G.2}, that 
 under a certain circumstance, the evolution of three $\alpha$-point vortices converges to a self-similar collapsing orbit in finite time as
 $\alpha \rightarrow 0$ and the variational part of the enstrophy dissipates in the sense of distributions at the event of collapse. In addition, it has also been revealed that
this is a singular mechanism that gives rise to the \textit{irreversibility of time} in conservative systems.
 
 Another important regularization appears in the numerical scheme to solve the 2D Euler equations, which is known as the \textit{vortex blob method} \cite{Anderson, Chorin, Krasny}.
 In this scheme, descretizing initial smooth vorticity distributions by a set of many point vortices, we approximate the evolution of the vorticity distributions
 with those of the point vortices, in which the regularized velocity field induced by a point vortex at ${\bm x}_0$ is given by
 \[
 {\bm u}^{\sigma}({\bm x}) = \frac{1}{2\pi}\frac{({\bm x}-{\bm x}_0)^\perp}{\vert {\bm x}-{\bm x}_0 \vert^2 + \sigma^2}.
\] 
Here, $\sigma$ denotes the smoothing parameter. As $\sigma \rightarrow 0$, we remark that the regularized velocity ${\bm u}^{\sigma}({\bm x})$ tends to a singular
velocity field that does not belong to $L^2_{loc}(\mathbb{R}^2)$. 

 Here arises a natural question which we are concerned with in the present paper: 
 Is the anomalous enstrophy dissipation via the triple collapse found in the Euler-$\alpha$ equations as $\alpha \rightarrow 0$  obtained similarly for 
 the flows regularized by the vortex blob method as $\sigma \rightarrow 0$? This is not only a theoretical
 extension of the preceding study~\cite{G.2}, but it should also be figured out whether or not the anomalous enstrophy dissipation can be constructed
 regardless of the regularization method.

As a matter of fact, these two regularizations of the incompressible velocity fields are generalized in a unified manner, which is called the \textit{Euler-Poincar\'{e}} (EP) 
system~\cite{Holm(a), Holm(b)}. It is derived from an application of Hamilton's principle to a dispersive kinetic energy  action function with a
 smoothing parameter $\varepsilon$. And the EP system is formally equivalent to the Euler equations when $\varepsilon$ is exactly zero.
The existence of the unique global solutions to the Euler-Poincar\'{e} equations for initial vorticity distributions in $\mathcal{M}(\mathbb{R}^2)$ has been established in \cite{G.3} 
as to the Euler-$\alpha$ equations. Furthermore, since the Euler-Poincar\'{e} equations share common mathematical structures with the Euler-$\alpha$ equations, 
the evolution of the weak solution can be investigated in terms of the dynamics of \textit{$\varepsilon$-point vortices}, which is introduced in this paper.
 Accordingly, one expects that the motion of the 
three $\varepsilon$-point vortices gives rise to the anomalous enstrophy dissipation via the self-similar collapse as we have shown in \cite{G.2}. On the other hand, in the 
Euler-Poincar\'{e} models, the enstrophy and the energy varying with the evolution of $\varepsilon$-point vortices are represented by Fourier transforms in terms of the smoothing 
function unlike the Euler-$\alpha$  equations where it is represented by elementary functions, which makes the mathematical treatment difficult.

 The paper is organized as follows. In Section~\ref{EP-eq}, we introduce the Euler-Poincar\'{e} equations,  and the existence and uniqueness theorem is stated. 
 We then give a mathematical formulation of the Euler-Poincar\'{e} point-vortex  (EP-PV) system in Section~\ref{EPPV} and its associated  enstrophy and energy variations 
  are defined in  Section~\ref{Variations}. After introducing the three $\varepsilon$-point vortex problem  in Section~\ref{three-vortex},  we summarize the main results in Section~\ref{Main},
  in which we provide a sufficient condition  for the emergence of the anomalous enstrophy dissipation. Sections~\ref{proof} and~\ref{proof-cor} are devoted to the proofs of the main results.  In Section~\ref{VBM}, we show that the enstrophy dissipation occurs for the flows regularized by the Gaussian kernel and by the vortex-blob method as applications of the main results.  Final section is concluding remarks. Appendix~\ref{functions} provides some properties of auxiliary functions, which play essential role in the proof of the main results.

\section{The Euler-Poincar\'{e} system }
\subsection{The Euler-Poincar\'{e} equations}
\label{EP-eq}
We derive a regularized 2D Euler equation, called the \textit{Euler-Poincar\'e} equation,  for incompressible velocity fields based on the framework of \cite{Foias(a), Holm(c)}. 
For an incompressible velocity field ${\bm v}$, let
us define ${\bm u}^\varepsilon$ by 
\begin{equation}
{\bm u}^\varepsilon({\bm x}) = \left( h^\varepsilon \ast {\bm v} \right) ({\bm x}) = \int_{\mathbb{R}^2} h^\varepsilon\left( {\bm x} - {\bm y} \right) {\bm v}({\bm y}) d{\bm y}, \label{Rvelo}
\end{equation}
in which a smoothing function $h^\varepsilon$ is  given by 
\begin{equation}
h^\varepsilon({\bm x}) = \frac{1}{\varepsilon^2} h\left( \frac{{\bm x}} {\varepsilon} \right)    \label{h_eps-h}
\end{equation}
for a scalar function $h({\bm x})$ on $\mathbb{R}^2$. We assume that $h$ is an integrable function and may have a singularity at the origin. 
Since ${\bm u}^\varepsilon$ is smoother than ${\bm v}$ owing to its definition, we call ${\bm u}^\varepsilon$ and ${\bm v}$ a \textit{regular velocity} and
 a \textit{singular velocity}, respectively. In a similar manner, we define the \textit{singular vorticity} $q$ and the \textit{regular vorticity} $\omega^\varepsilon$ by
$q = \operatorname{curl} {\bm v}$ and $\omega^\varepsilon = \operatorname{curl} {\bm u}^\varepsilon$.
Let us remark that $\operatorname{div} {\bm u}^\varepsilon = 0$ and $\omega^\varepsilon = h^\varepsilon \ast q$ when the convolution commutes with the differential operator.
 Then the Euler-Poincar\'{e} equations for $({\bm u}^\varepsilon, {\bm v})$ in $\mathbb{R}^2$ are given by
 \begin{equation}
\partial_t {\bm v} + ({\bm u}^\varepsilon \cdot \nabla) {\bm v} - (\nabla {\bm v})^T \cdot {\bm u}^\varepsilon - \nabla \Pi = 0,\qquad  \operatorname{div} {\bm u}^\varepsilon = \operatorname{div} {\bm v}=0, \label{REE}
\end{equation}
where $\Pi$ is a generalized pressure. The first momentum equation is derived from the Hamiltonian
\begin{equation*}
\mathscr{H} = \frac{1}{2} \int_{\mathbb{R}^2} {\bm v}({\bm x}) \cdot {\bm u}^\varepsilon({\bm x}) d{\bm x}
\end{equation*}
subject to the divergence-free condition through Hamilton's principle. Taking the curl of (\ref{REE}), we obtain the transport equations 
for the singular vorticity advected by the regular velocity: 
\begin{equation}
\partial_t q + ({\bm u}^\varepsilon \cdot \nabla) q = 0, \qquad {\bm u}^\varepsilon = {\bm K}^\varepsilon \ast q, \qquad {\bm K}^\varepsilon = {\bm K} \ast h^\varepsilon,  \label{RPDE}
\end{equation}
where  ${\bm u}^\varepsilon = {\bm K} \ast \omega^\varepsilon$ owing to the Biot-Savart formula.
It has been shown in \cite{G.3} that the initial value problem of (\ref{RPDE}) has a unique global weak solution in the space of Radon measures $\mathcal{M}(\mathbb{R}^2)$ on
$\mathbb{R}^2$,
in which the following equation for the Lagrangian flow map ${\bm \eta}^\varepsilon$  associated with the regular velocity is considered.
\begin{equation}
\partial_t {\bm \eta}^\varepsilon({\bm x}, t) = {\bm u}^\varepsilon \left( {\bm \eta}^\varepsilon({\bm x}, t), t \right), \qquad  {\bm \eta}^\varepsilon({\bm x}, 0) = {\bm x}. \label{EP-flow}
\end{equation}
The solution of (\ref{EP-flow}) yields that of (\ref{RPDE}) with the initial vorticity $q_0$ as follows.
\begin{equation}
q({\bm x}, t) = q_0 \left( {\bm \eta}^\varepsilon({\bm x}, - t) \right). 
\end{equation}
Here, the following functions  are introduced to characterize singularities and decay rates of functions. 
\begin{equation*}
\chi_{\log}^{-} ({\bm x})  = \left\{
\begin{array}{cc}
\displaystyle{ \left( 1 - \log{\vert{\bm x}\vert} \right)^{-1} }  &, \ \vert{\bm x}\vert \leq 1 , \\
\displaystyle{ 0 }  &, \ \vert{\bm x}\vert > 1,
\end{array}
\right. \quad 
\chi^{+}_{\alpha}({\bm x}) = \left\{
\begin{array}{cc}
0   &,  \ \vert{\bm x}\vert \leq 1, \\
\vert{\bm x}\vert^\alpha  &,  \ \vert{\bm x}\vert > 1 .
\end{array}
\right. 
\end{equation*}
We also set $\chi_{\alpha}({\bm x}) = \vert{\bm x}\vert^\alpha$ for ${\bm x}\in \mathbb{R}^2$. Then the following theorem holds.

\begin{theorem} (\textbf{\cite{G.3}})
Suppose that $h \in C^1(\mathbb{R}^2) \cap W^1_1(\mathbb{R}^2)$ satisfies $\chi_1^{+} h \in L^1(\mathbb{R}^2)$ and
\begin{equation}
\chi_{\log}^{-} h \in L^\infty(\mathbb{R}^2), \qquad  \chi_1 \nabla h \in L^\infty(\mathbb{R}^2). \label{h-asympt}
\end{equation}
Then, for any initial vorticity $q_0 \in \mathcal{M}(\mathbb{R}^2)$, there exists a unique global weak solution of (\ref{RPDE}) such that 
${\bm \eta}^\varepsilon \in C^1(\mathbb{R};\mathscr{G})$, ${\bm u}^\varepsilon \in C(\mathbb{R}; C (\mathbb{R}^2; \mathbb{R}^2))$ and $q \in C(\mathbb{R};\mathcal{M}(\mathbb{R}^2))$,
where $\mathscr{G}$ denotes the group of homeomorphisms on $\mathbb{R}^2$ that preserve the Lebesgue measure. \label{well-posed}
\end{theorem}

\begin{remark}
The assumptions of Theorem~\ref{well-posed} are satisfied with two well-known regularization of the Euler equations;  the Euler-$\alpha$ equations
for $h({\bm x})=K_0(\vert{\bm x}\vert)/(2\pi)$ and the vortex blob method for $h({\bm x})=1/(\pi(\vert{\bm x}\vert^2+1)^{-2})$. See \cite{G.3, Holm(c)}.
\end{remark}

\subsection{The Euler-Poincar\'{e} point vortex system}
\label{EPPV}
In what follows, we suppose that the smoothing function $h$ in (\ref{h_eps-h}) satisfies the assumptions of Theorem~\ref{well-posed}. In addition, suppose that $h$ is radial, 
namely $h_r(\vert{\bm x}\vert) = h({\bm x})$, and it satisfies
\begin{equation}
\int_{\mathbb{R}^2} h({\bm x}) d{\bm x} = 2 \pi \int_0^\infty r h_r (r) dr = 1. \label{h^eps-1}
\end{equation}
We first investigate the properties of ${\bm K}^\varepsilon$. As shown in \cite{G.3}, under the assumptions of Theorem~\ref{well-posed}, ${\bm K}^\varepsilon$ belongs to
 $C_0(\mathbb{R}^2)$ and it is quasi-Lipschitz continuous with ${\bm K}^\varepsilon({\bm 0}) = 0$. It is also important to remark that ${\bm K}^\varepsilon$ is defined by
 ${\bm K}^\varepsilon = \nabla^\perp G^\varepsilon$,
 where $G^\varepsilon$ is a solution to the following Poisson equation for $h^\varepsilon$:
\begin{equation}
- \Delta G^\varepsilon = h^\varepsilon.
\label{Poisson_G}
\end{equation}
If $h$ is radial, so is $G^\varepsilon$, say $G^\varepsilon({\bm x}) = G_r^\varepsilon(\vert{\bm x}\vert)$ and we have the relation,
\begin{equation}
G^\varepsilon({\bm x}) = G^1\left( \frac{{\bm x}}{\varepsilon} \right) - \frac{1}{2 \pi}\log{\varepsilon}. 
\label{G_1}
\end{equation}
Then, we have
\begin{equation}
{\bm K}^\varepsilon({\bm x}) = \frac{{\bm x}^\perp}{\varepsilon\vert{\bm x}\vert} \frac{\mbox{d}G_r^1}{\mbox{d}r} \left( \frac{\vert{\bm x}\vert}{\varepsilon} \right) \equiv {\bm K}({\bm x}) P_K \left( \frac{\vert{\bm x}\vert}{\varepsilon} \right),  \label{K^eps-KP}
\end{equation}
where $P_K(r)$ is defined by
\begin{equation}
P_K(r) = - 2 \pi r \frac{\mbox{d}G_r^1}{\mbox{d}r} (r).
\label{P_K}
\end{equation}

Suppose now that the initial vorticity field is represented by a set of $\delta$-distributions, 
\begin{equation}
q_0({\bm x}) = \sum_{n=1}^N \Gamma_n \delta({\bm x} - {\bm x}_n^0), \label{q0}
\end{equation}
where ${\bm x}_n^0 = (x_n^0, y_n^0) \in \mathbb{R}^2$ for $n=1,\dots,N$ are their point supports of the $\delta$-singularities, called 
\textit{$\varepsilon$-point vortices}. The strength $\Gamma_n \in \mathbb{R}$ corresponds to the circulation around the $\varepsilon$-point vortex at ${\bm x}_n^0$.
Theorem~\ref{well-posed} shows that there exists a unique global weak solution to (\ref{RPDE}) with the initial data (\ref{q0}). 
More precisely, we have the following proposition.

\begin{proposition}
Suppose that $h$ satisfies the assumptions of Theorem~\ref{well-posed}. Then, the solution to (\ref{RPDE}) with the initial data (\ref{q0}) is expressed by
\begin{equation}
q({\bm x}, t) = \sum_{n=1}^N \Gamma_n \delta( {\bm x} - {\bm \eta}^\varepsilon({\bm x}_n^0, t ) ).  \label{q_sol}
\end{equation}
Moreover, the point vortices in the Euler-Poincar\'{e} system never collapse. \label{pv-solution}
\end{proposition}

\begin{proof}
Since Theorem~\ref{well-posed} assures the  existence of a unique global solution to (\ref{EP-flow}), we have
\begin{equation*}
q({\bm x}, t) = q_0({\bm \eta}^\varepsilon({\bm x}, -t)) = \sum_{n=1}^N \Gamma_n \delta( {\bm \eta}^\varepsilon({\bm x}, -t) - {\bm x}_n^0 ).
\end{equation*}
If ${\bm \eta}^\varepsilon({\bm x}, -t) = {\bm x}_n^0$ then ${\bm x} = {\bm \eta}^\varepsilon({\bm x}_n^0, t)$ else ${\bm x} \neq {\bm \eta}^\varepsilon({\bm x}_n^0, t)$, which implies
$\delta( {\bm \eta}^\varepsilon({\bm x}, -t) - {\bm x}_n^0 ) = \delta( {\bm x} - {\bm \eta}^\varepsilon({\bm x}_n^0, t) )$.
Moreover, it follows from the uniqueness of the flow map that ${\bm \eta}^\varepsilon({\bm x}_m^0, t) \neq {\bm \eta}^\varepsilon({\bm x}_n^0, t)$ for $m \neq n$ and an arbitrary $t \in \mathbb{R}$. Thus, there is no collapse. 
\end{proof}

The evolution of $\varepsilon$-point vortices is described by
 ${\bm x}_n^\varepsilon(t) = {\bm \eta}^\varepsilon({\bm x}_n^0, t)$. It follows from Proposition~\ref{pv-solution} and (\ref{K^eps-KP}) with $K^\varepsilon({\bm 0})=0$, the equations (\ref{EP-flow}) with the initial vorticity (\ref{q0}) are
 equivalent to
\begin{equation}
\frac{\mbox{d}}{\mbox{d}t} {\bm x}_n^\varepsilon(t) = {\bm u}^\varepsilon\left( {\bm x}_n^\varepsilon(t), t \right) = - \frac{1}{2\pi} \sum_{m\neq n}^N \Gamma_m \frac{( {\bm x}_n^\varepsilon - {\bm x}_m^\varepsilon )^\perp}{(l_{mn}^\varepsilon) ^2} P_K \left( \frac{l_{mn}^\varepsilon}{\varepsilon} \right), \quad  n = 1,\dots,N,   \label{xEP-PV}  
\end{equation}
where $l_{mn}^\varepsilon(t) = \vert{\bm x}_n^\varepsilon(t) - {\bm x}_m^\varepsilon(t)\vert$ and ${\bm x}_n^\varepsilon(0) = {\bm x}_n^0$.
 The evolution equation for $\varepsilon$-point vortices is called  the Euler-Poincar\'{e} point vortex (EP-PV) system.
According to Proposition~\ref{pv-solution}, a weak solution to the 2D Euler-Poincar\'{e} equations provides a solution of the EP-PV system and vice versa. 
Now, let us see some  properties of the EP-PV system. Considering the relation
\begin{equation*}
G_r^\varepsilon(\vert{\bm x}\vert) = - \frac{1}{2 \pi} \left[ \log{\vert{\bm x}\vert} + H_G\left( \frac{\vert{\bm x}\vert}{\varepsilon} \right) \right]
\end{equation*}
with 
\begin{equation}
H_G(r) = - \log{r} - 2 \pi G_r^1(r), 
\label{H_G}
\end{equation}
we find that (\ref{xEP-PV}) is formulated as a Hamiltonian dynamical system. That is to say, it is equivalent to
\begin{equation*}
\Gamma_n \frac{\mbox{d} x_n^\varepsilon}{\mbox{d}t}  = \frac{\partial \mathscr{H}^\varepsilon}{\partial y_n^\varepsilon}, \qquad \Gamma_n \frac{\mbox{d} y_n^\varepsilon}{\mbox{d}t}  = - \frac{\partial \mathscr{H}^\varepsilon}{\partial x_n^\varepsilon}, \qquad  n = 1,\dots,N,
\end{equation*}
with the Hamiltonian
\begin{equation}
\mathscr{H}^\varepsilon = - \frac{1}{2 \pi} \sum_{n=1}^N \sum_{m=n+1}^N \Gamma_n \Gamma_m \left[ \log{l_{mn}^\varepsilon} + H_G\left( \frac{l_{mn}^\varepsilon}{\varepsilon} \right) \right]. \label{EP-Hamiltonian}
\end{equation}
The EP-PV system (\ref{xEP-PV}) admits four conserved quantities $(\mathscr{H}^\varepsilon, Q^\varepsilon, P^\varepsilon, I^\varepsilon)$, where 
\begin{equation*}
Q^\varepsilon + i P^\varepsilon = \sum_{n = 1}^N x_n^\varepsilon + i y_n^\varepsilon(t), \qquad I^\varepsilon = \sum_{n = 1}^N \Gamma_n \left[ (x_n^\varepsilon)^2 + (y_n^\varepsilon)^2 \right].
\end{equation*}
We then have the following integrability of the EP-PV system.
\begin{proposition}
 The EP-PV system (\ref{xEP-PV}) for $N \leq 3$ is integrable for any strengths of point vortices. It is also integrable for $N = 4$ when the total vortex strength is zero, 
 i.e. $\Gamma = \sum_{n = 1}^N \Gamma_n = 0$.
\end{proposition}
\begin{proof}
Defining the Poisson bracket between two functions $f$ and $g$ by
\begin{equation*}
\{ f, g \} = \sum_{n = 1}^N \frac{1}{\Gamma_n} \left( \frac{\partial f}{\partial x_n}\frac{\partial g}{\partial y_n} - \frac{\partial f}{\partial y_n}\frac{\partial g}{\partial x_n} \right),
\end{equation*}
we find $\{ \mathscr{H}^\varepsilon, I^\varepsilon \} = 0$, $\{ \mathscr{H}^\varepsilon, (P^\varepsilon)^2 + (Q^\varepsilon)^2 \} = 0$ and $\{ (P^\varepsilon)^2 + (Q^\varepsilon)^2, I^\varepsilon \} = 0$.  In addition to above invariant quantities, we have $\{ Q^\varepsilon, P^\varepsilon \} = \Gamma$, $\{ Q^\varepsilon, I^\varepsilon \} = 2 P^\varepsilon$ and $\{ P^\varepsilon, I^\varepsilon \} = - 2 Q^\varepsilon$. 
\end{proof}
In the case $N = 4$ with $\Gamma \neq 0$ or $N \geq 5$, the system is no longer integrable and the dynamics of $\varepsilon$-point vortices could be chaotic. 
Another important conserved quantity $M^\varepsilon$ is introduced  by
\begin{equation*}
M^\varepsilon = \sum_{n \neq m}^N \Gamma_n \Gamma_m (l_{mn}^\varepsilon)^2 = 2 (\Gamma I^\varepsilon - (Q^\varepsilon)^2 - (P^\varepsilon)^2),
\end{equation*}
which depends only on the distances $l_{mn}^\varepsilon$ between two $\varepsilon$-point vortices at ${\bm x}_m^\varepsilon$ and ${\bm x}_n^\varepsilon$.

\subsection{Variations of energy and enstrophy}
\label{Variations}
We are concerned with  the enstrophy and the energy varying with the evolution of $\varepsilon$-point vortices, which are
 derived based on the Novikov's method \cite{Novikov, Sakajo}. We define the Fourier transform of the function $f$ by
\begin{equation}
\mathscr{F}[f]({\bm k}) = \frac{1}{2 \pi}\int_{\mathbb{R}^2} f({\bm x}) e^{- i {\bm x} \cdot {\bm k}} d {\bm x}. \label{FT}
\end{equation}
Note that if $f$ is radial, i.e., $f_r(\vert{\bm x}\vert) = f({\bm x})$, then its Fourier transform is equivalent to the Hankel transform of $f_r$, 
\begin{equation}
\widehat{f}(s) = \mathscr{F}[f]({\bm k}) = \int_0^\infty r f_r(r) J_0(rs) dr, \label{HT}
\end{equation}
in which $r = \vert{\bm x}\vert$ and $s = \vert{\bm k}\vert$.  First, the total enstrophy for the regular vorticity is given by
\begin{equation*}
\frac{1}{2}\int_{\mathbb{R}^2} \left\vert \omega^\varepsilon({\bm x}, t) \right\vert^2 d{\bm x} = \frac{1}{2}\int_{\mathbb{R}^2} \left\vert \mathscr{F}[\omega^\varepsilon]({\bm k}, t) \right\vert^2 d{\bm k} = \int_0^\infty  \pi s \langle \left\vert \mathscr{F}[\omega^\varepsilon](s,t) \right\vert^2 \rangle ds,
\end{equation*} 
where $\langle f \rangle = \frac{1}{2\pi}\int_{-\pi}^\pi f(\theta) d\theta$. Here, we define the enstrophy density spectrum $\mathscr{Z}_N^\varepsilon$ by 
\begin{equation*}
\mathscr{Z}_N^\varepsilon (s,t) =  \pi s \langle \left\vert \mathscr{F}[\omega^\varepsilon](s,t) \right\vert^2 \rangle. 
\end{equation*}
The Fourier transform of the vorticity field (\ref{q_sol}) is represented by
\begin{equation*}
\mathscr{F}[q]({\bm k}, t) = \frac{1}{2 \pi} \sum_{n=1}^N \Gamma_n e^{- i {\bm k}\cdot {\bm x}_n^\varepsilon(t)}.
\end{equation*}
Hence, we have
\begin{equation*}
\left\vert \mathscr{F}[q]({\bm k}, t) \right\vert^2 = \frac{1}{4 \pi^2} \left[ \sum_{n=1}^N \Gamma_n^2 + 2 \sum_{n=1}^N \sum_{m=n+1}^N \Gamma_n \Gamma_m \cos \left( {\bm k} \cdot ({\bm x}_n^\varepsilon(t) - {\bm x}_m^\varepsilon(t)) \right)  \right].
\end{equation*}
Since $h^\varepsilon$ is radial and $\mathscr{F}[\omega^\varepsilon] = 2 \pi \mathscr{F}[h^\varepsilon] \mathscr{F}[q]$ owing to $\omega^\varepsilon = h^\varepsilon \ast q$, we obtain
\begin{align*}
\mathscr{Z}_N^\varepsilon (s,t) &= \frac{s}{4 \pi} \left\vert 2 \pi \widehat{h^\varepsilon}(s) \right\vert^2 \left[ \sum_{n=1}^N \Gamma_n^2 + 2 \sum_{n=1}^N \sum_{m=n+1}^N \Gamma_n \Gamma_m \frac{1}{2\pi}\int_{-\pi}^\pi \cos \left( s l_{mn}^\varepsilon \cos \theta \right) d\theta \right]  \\
& = \frac{s}{4 \pi} \left\vert 2 \pi \widehat{h^\varepsilon}(s) \right\vert^2 \left[ \sum_{n=1}^N \Gamma_n^2 + 2 \sum_{n=1}^N \sum_{m=n+1}^N \Gamma_n \Gamma_m J_0 \left( s l_{mn}^\varepsilon \right) \right],
\end{align*}
where $J_0(s)$ is a Bessel function of the first kind. Accordingly, the total enstrophy for the EP-PV system is expressed by
\begin{align*}
&\frac{1}{2}\int_{\mathbb{R}^2} \left\vert \omega^\varepsilon({\bm x}, t) \right\vert^2 d{\bm x} = \int_0^\infty  \mathscr{Z}_N^\varepsilon (s,t)  ds  \\
&= \frac{1}{4 \pi \varepsilon^2} \sum_{n=1}^N \Gamma_n^2 \int_0^\infty s \left\vert 2\pi\widehat{h}(s) \right\vert^2 ds  +  \frac{1}{2\pi \varepsilon^2} \sum_{n=1}^N \sum_{m=n+1}^N \Gamma_n \Gamma_m \int_0^\infty s \left\vert 2\pi\widehat{h}(s) \right\vert^2  J_0 \left( s \frac{l_{mn}^\varepsilon}{\varepsilon} \right) ds.
\end{align*} 
Here, we use the relation $\widehat{h^\varepsilon}(s) = \widehat{h}(\varepsilon s)$. Since the first term in the right-hand side is constant in time,  the variational part of the enstrophy  is provided 
by the second term, namely, 
\begin{equation}
\mathscr{Z}^\varepsilon(t) \equiv \frac{1}{2\pi \varepsilon^2} \sum_{n=1}^N \sum_{m=n+1}^N \Gamma_n \Gamma_m \int_0^\infty s \left\vert 2\pi\widehat{h}(s) \right\vert^2  J_0 \left( s \frac{l_{mn}^\varepsilon(t)}{\varepsilon} \right) ds. \label{v-enstrophy}
\end{equation}
Second,  the total energy for the regular velocity is defined by
\begin{equation*}
\frac{1}{2}\int_{\mathbb{R}^2} \left\vert {\bm u}^\varepsilon({\bm x}, t) \right\vert^2 d{\bm x} = \int_0^\infty  \pi s \langle \left\vert \mathscr{F}[{\bm u}^\varepsilon](s,t) \right\vert^2 \rangle ds.
\end{equation*}
Since $\left\vert \mathscr{F}[\omega^\varepsilon]({\bm k}, t) \right\vert^2 = \vert{\bm k}\vert^2 \left\vert \mathscr{F}[{\bm u}^\varepsilon]({\bm k}, t) \right\vert^2$,  the energy density spectrum $E_N^\varepsilon$ is represented by
\begin{align*}
E_N^\varepsilon(s, t) &\equiv \pi s \langle \left\vert \mathscr{F}[{\bm u}^\varepsilon](s,t) \right\vert^2 \rangle = \frac{\pi}{s} \langle \left\vert \mathscr{F}[\omega^\varepsilon](s,t) \right\vert^2 \rangle  \\
& = \frac{1}{4 \pi s} \left\vert 2\pi\widehat{h^\varepsilon}(s) \right\vert^2 \left[ \sum_{n=1}^N \Gamma_n^2 + 2 \sum_{n=1}^N \sum_{m=n+1}^N \Gamma_n \Gamma_m J_0 \left( s l_{mn}^\varepsilon \right) \right]. 
\end{align*}
Hence, the total energy that is cut off at a scale larger than $1 \ll L < \infty$  is expressed by
\begin{align}
\int_{L^{-1}}^\infty E_N^\varepsilon(s, t) dr &= \frac{1}{4 \pi} \sum_{n=1}^N \Gamma_n^2\int_{\varepsilon L^{-1}}^\infty \frac{1}{s} \left\vert 2\pi\widehat{h}(s) \right\vert^2 ds   \nonumber \\
& \quad +  \frac{1}{2\pi} \sum_{n=1}^N \sum_{m=n+1}^N \Gamma_n \Gamma_m \int_{\varepsilon L^{-1}}^\infty \frac{1}{s} \left\vert 2\pi\widehat{h}(s) \right\vert^2  J_0 \left( s \frac{l_{mn}^\varepsilon}{\varepsilon} \right) ds. \label{TotalE}
\end{align} 
The second term is rewritten as follows.
\[
\int_{\varepsilon L^{-1}}^\infty \frac{1}{s} \left\vert 2\pi\widehat{h}(s) \right\vert^2  J_0 \left( s \frac{l_{mn}^\varepsilon}{\varepsilon} \right) ds = \int_{\varepsilon L^{-1}}^\infty \frac{1}{s} J_0 \left( s \frac{l_{mn}^\varepsilon}{\varepsilon} \right) ds + \int_{\varepsilon L^{-1}}^\infty \frac{1}{s} \left( \left\vert 2\pi\widehat{h}(s) \right\vert^2 - 1 \right) J_0 \left( s \frac{l_{mn}^\varepsilon}{\varepsilon} \right) ds.
\]
Then, the following approximation holds for a sufficiently large $L$.
\begin{align*}
\int_{\varepsilon L^{-1}}^\infty \frac{1}{s} J_0 \left( s \frac{l_{mn}^\varepsilon}{\varepsilon} \right) ds &\sim \int_0^\infty \frac{ s }{s^2 + (\varepsilon L^{-1})^2} J_0 \left( s \frac{l_{mn}^\varepsilon}{\varepsilon}\right)  ds  = K_0\left( \frac{l_{mn}^\varepsilon}{L} \right) \\
& \sim - \log l_{mn}^\varepsilon  + \log{\frac{L e^\beta}{2}} + \mathcal{O}\left( L^{-2} \log{L^{-1}} \right), 
\end{align*}
in which $\beta$ is the Euler's constant. Taking the $L \rightarrow \infty$ limit in the non-constant part of the total energy (\ref{TotalE}), we obtain the variational part of the energy as follows.
\begin{equation}
E^\varepsilon(t) \equiv - \frac{1}{2\pi} \sum_{n=1}^N \sum_{m=n+1}^N \Gamma_n \Gamma_m \left[ \log{l_{mn}^\varepsilon(t)} +  \int_0^\infty \frac{1}{s} \left( 1 - \left\vert 2\pi\widehat{h}(s) \right\vert^2 \right) J_0 \left( s \frac{l_{mn}^\varepsilon(t)}{\varepsilon} \right) ds \right]. \label{v-energy}
\end{equation}
Note that the integrand of the second term has no singularity, since it follows from (\ref{h^eps-1}) that $2\pi \widehat{h}(0) =  1$.


\section{Main results}
\subsection{The three $\varepsilon$-vortex problem}
\label{three-vortex}

We consider the three $\varepsilon$-point vortex problem, i.e. the EP-PV system with $N=3$, whose Hamiltonian is expressed explicitly by
\begin{align*}
\mathscr{H}^\varepsilon = & -\frac{1}{2\pi} \left( \Gamma_2 \Gamma_3 \log l_{23}^\varepsilon + \Gamma_3 \Gamma_1 \log l_{31}^\varepsilon + \Gamma_1 \Gamma_2 \log l_{12}^\varepsilon \right) \\ 
& - \frac{1}{2\pi} \left[ \Gamma_2 \Gamma_3 H_G\left( \frac{l_{23}^\varepsilon}{\varepsilon} \right) + \Gamma_3 \Gamma_1 H_G\left( \frac{l_{31}^\varepsilon}{\varepsilon} \right) + \Gamma_1 \Gamma_2 H_G\left( \frac{l_{12}^\varepsilon}{\varepsilon} \right) \right].  
\end{align*}
It is reduced to the following equation for the distance $l_{mn}^\varepsilon$:
\begin{equation}
\frac{\mbox{d}}{\mbox{d}t} (l_{mn}^\varepsilon)^2 = \frac{2}{\pi} \Gamma_k A^\varepsilon \left[ \frac{1}{(l_{nk}^\varepsilon)^2} P_K \left( \frac{l_{nk}^\varepsilon}{\varepsilon} \right) - \frac{1}{(l_{km}^\varepsilon)^2} P_K \left( \frac{l_{km}^\varepsilon}{\varepsilon} \right) \right], \quad l_{mn}^\varepsilon(0) = \vert{\bm x}_m^0 - {\bm x}_n^0\vert, \label{l-EPPV}
\end{equation}
where $k, m, n \in \{ 1,2,3 \}$ with $k \neq m \neq n$ and $A^\varepsilon$ denotes the signed area of the triangle formed by the three $\varepsilon$-point vortices. Its sign is positive if ${\bm x}_1^\varepsilon$, ${\bm x}_2^\varepsilon$ and ${\bm x}_3^\varepsilon$ at the vertices of the triangle appear counterclockwise, while it is negative if they do clockwise. Remember that we have the two invariants in terms of the lengths;  $\mathscr{H}^\varepsilon$ and 
\begin{equation*}
M^\varepsilon = \Gamma_2 \Gamma_3 (l_{23}^\varepsilon)^2 + \Gamma_3 \Gamma_1 (l_{31}^\varepsilon)^2 + \Gamma_1 \Gamma_2 (l_{12}^\varepsilon)^2. 
\end{equation*}
In order to take the $\varepsilon \rightarrow 0$ limit, we introduce the following scaled variables:
\begin{equation}
{\bm X}_n(t) = \frac{1}{\varepsilon} {\bm x}_n^\varepsilon(\varepsilon^2 t + t^\ast ), \qquad L_{mn}(t) =  \frac{1}{\varepsilon} l_{mn}^\varepsilon (\varepsilon^2 t + t^\ast)  \label{scaled}
\end{equation}
for $m, n = \{ 1, 2, 3 \}$ with $m \neq n$, where $t^\ast\in \mathbb{R}$ is an arbitrary constant determined later. Then, the evolution equation for ${\bm X}_n(t)$ is described by 
\begin{equation}
\frac{\mbox{d}}{\mbox{d}t} {\bm X}_n = -\frac{1}{2\pi} \sum_{m\neq n}^3 \Gamma_m \frac{( {\bm X}_n - {\bm X}_m )^\perp}{L_{mn}^2} P_K\left( L_{mn} \right), \qquad {\bm X}_n(0) = \frac{{\bm x}_n^\varepsilon(t^\ast)}{\varepsilon}.  \label{XEP-PV}
\end{equation}
It  is also a Hamiltonian system with the Hamiltonian 
\begin{equation}
\mathscr{H} = \mathscr{H}^1 = - \frac{1}{2\pi} \left[ \Gamma_2 \Gamma_3 H_P\left( L_{23}^2 \right) + \Gamma_3 \Gamma_1 H_P\left( L_{31}^2 \right) + \Gamma_1 \Gamma_2 H_P\left( L_{12}^2 \right) \right],  \label{H}
\end{equation}
where 
\begin{equation}
H_P(r) = \log{\sqrt{r}} + H_G \left( \sqrt{r} \right),
\label{H_P}
\end{equation}
and it has an invariant quantity, 
\begin{equation}
M = M^1 = \Gamma_2 \Gamma_3 L_{23}^2 + \Gamma_3 \Gamma_1 L_{31}^2 + \Gamma_1 \Gamma_2 L_{12}^2.  \label{M}
\end{equation}
The evolution of the distance $L_{mn}$ is governed by
\begin{equation}
\frac{\mbox{d}}{\mbox{d}t} L_{mn}^2 = \frac{2}{\pi} \Gamma_k A \left[ \frac{1}{L_{nk}^2} P_K \left( L_{nk} \right) - \frac{1}{L_{km}^2} P_K \left( L_{km} \right) \right], \label{L_123} \\
\end{equation}
where $A = A^1$ is the signed area of the triangle formed by the three $\varepsilon$-point vortices at ${\bm X}_1(t)$, ${\bm X}_2(t)$ and ${\bm X}_3(t)$. We easily observe that the relative equilibria of (\ref{L_123}) are equilateral triangles or collinear configurations.  See Proposition~1 of \cite{G.2} for its proof.

We remark that the solutions of (\ref{xEP-PV}) and (\ref{l-EPPV}) are recovered from those of the scaled systems (\ref{XEP-PV}) and (\ref{L_123}) via
\begin{equation}
{\bm x}_n^\varepsilon(t) = \varepsilon {\bm X}_n \left(\frac{t - t^\ast}{\varepsilon^2} \right), \qquad l_{mn}^\varepsilon (t) = \varepsilon L_{mn} \left( \frac{t - t^\ast}{\varepsilon^2} \right).  \label{recover}
\end{equation}
In terms of the scaled variables, the variation of enstrophy $\mathscr{Z}^\varepsilon(t)$ is described by 
\begin{equation}
\mathscr{Z}^\varepsilon(t) = - \frac{1}{\varepsilon^2} \mathscr{Z}_0\left( \frac{t - t^\ast}{\varepsilon^2} \right), \qquad \mathscr{Z}_0(\tau) = - \frac{1}{2\pi} \sum_{n=1}^N \sum_{m=n+1}^N \Gamma_n \Gamma_m Z_{mn}(\tau),
\label{Z_0}
\end{equation}
in which 
\begin{equation}
Z_{mn}(\tau) = \int_0^\infty s \left\vert 2\pi\widehat{h}(s) \right\vert^2  J_0 \left( s L_{mn}(\tau) \right) ds. \label{Z_nm}
\end{equation}
Regarding the energy variation $E^\varepsilon(t)$, since $\mathscr{H}^\varepsilon$ is constant, we rewrite (\ref{v-energy}) as follows.
\begin{equation*}
E^\varepsilon(t) = \mathscr{H}^\varepsilon +  E_0\left( \frac{t - t^\ast}{\varepsilon^2} \right), \qquad E_0(\tau) = - \frac{1}{2\pi} \sum_{n=1}^N \sum_{m=n+1}^N \Gamma_n \Gamma_m E_{mn}(\tau),
\end{equation*}
in which
\begin{equation*}
E_{mn}(\tau) = - H_G\left( L_{mn}(\tau) \right) +  \int_0^\infty \frac{1}{s} \left( 1 - \left\vert 2\pi\widehat{h}(s) \right\vert^2 \right) J_0 \left( s L_{mn}(\tau) \right) ds.
\end{equation*}
The energy dissipation rate $\mathscr{D}_E^\varepsilon$ is obtained by differentiating $E^\varepsilon(t)$:
\begin{equation*}
\mathscr{D}_E^\varepsilon (t) = \frac{1}{\varepsilon^2} \frac{\mbox{d}E_0}{\mbox{d}\tau} \left( \frac{t - t^\ast}{\varepsilon^2} \right).  
\end{equation*}

\begin{remark}
In summary, from a given radial smoothing function $h({\bm x})=h_r(\vert {\bm x} \vert )$, the functions $P_K( \vert {\bm x} \vert)$, $H_G(\vert{\bm x}\vert)$, $H_P(\vert {\bm x} \vert)$
and $L_P(\vert {\bm x} \vert)$ are derived as follows. For $h^\varepsilon$ defined by (\ref{h_eps-h}),  the function $G_r^\varepsilon( \vert {\bm x} \vert)$ is
obtained as the solution of the Poisson equation (\ref{Poisson_G}). Setting $\varepsilon=1$, we have $G_r^1(\vert{\bm x}\vert)$, yielding the functions $P_K$ and $H_G$
as (\ref{P_K}) and (\ref{H_G}), respectively. According to (\ref{H_P}), the function $H_P$ is derived from $H_G$. The function $L_P(\vert{\bm x}\vert)$ is defined by
\begin{equation}
L_P(\vert{\bm x}\vert) = \frac{1}{\vert {\bm x}\vert} P_K\left(\sqrt{\vert {\bm x} \vert}\right).
\label{L_P}
\end{equation}
Note that $L_P$ also satisfies
\begin{equation}
L_P(r) = 2 \frac{\mbox{d}}{\mbox{d}r}H_P(\sqrt{r}).
\label{dL_p}
\end{equation}
Those functions play a significant role in the investigation of the dynamics of the three $\varepsilon$-point vortices  shown later. 

The EP-PV system is a generalization of the $\alpha$PV system derived from the Euler-$\alpha$ equations considered in \cite{G.2}. 
 We remark that  for the modified Bessel function $K_0(\vert{\bm x}\vert)$ as the smoothing function in  the $\alpha$PV system, the functions $P_K(\vert{\bm x}\vert)$, $H_G(\vert{\bm x}\vert)$, $H_P(\vert{\bm x}\vert)$ and $L_p(\vert{\bm x}\vert)$ are correspondingly  denoted by  $B_K(\vert {\bm x} \vert)$, 
 $K_0(\vert{\bm x}\vert)$, $h_0(\vert{\bm x}\vert)$ and $h_K(\vert{\bm x}\vert)$ in the paper \cite{G.2}. Table~\ref{Tab_Function} is a summary of the correspondence between those functions and their common properties, whose proofs are  provided in Appendix~\ref{functions}.
 \end{remark}
 \begin{table}
 \centering
 \begin{tabular}{|c|c|c|} \hline
 $h_r(r)$ &  $K_0(r)$ & properties \\ \hline
 $P_K(r)$ & $B_K(r)$ & monotone increasing, upward convex, $0 \leq P_K < 1$ \\
 $H_G(r)$ & $K_0(r)$ & monotone decreasing, downward convex, (\ref{H_G-zero}) \\ 
  $H_P(r)$ & $h_0(r)$ & monotone increasing, upward convex  \\
  $L_P(r)$ & $h_K(r)$ & positive, monotone decreasing \\ \hline
 \end{tabular}
 \caption{Functions associated with the smoothing function $h_r(r)$ for the EP-PV system corresponding to $K_0(r)$ for the $\alpha$PV system. Their common properties used in 
 the analysis of the three vortex problem are also provided.}
 \label{Tab_Function}
 \end{table}

\subsection{Main theorems}
\label{Main}
As  in \cite{G.2}, the existence of the enstrophy dissipating solution of the three $\alpha$PV system is shown under the condition
\begin{equation}
\frac{1}{\Gamma_1} + \frac{1}{\Gamma_2} + \frac{1}{\Gamma_3} = 0.  \label{condi1} 
\end{equation}
In view of (\ref{condi1}), we may assume $\Gamma_1 \geq \Gamma_2 > 0 > \Gamma_3$ without loss of generality. Note that that (\ref{condi1}) yields $\mathscr{H}^\varepsilon = \mathscr{H}$. We then show the existence of the evolution of the three $\varepsilon$-point vortices whose enstrophy varies and energy is conserved in the sense of distributions in the $\varepsilon \rightarrow 0$ limit. To state the theorem, we introduce the following functions that are defined only from the strengths 
$\Gamma_1$ and $\Gamma_2$.
\begin{equation}
\psi(r) = \left( \frac{1}{1+r} \right)^{1/\Gamma_1}  \left( \frac{r}{1+r} \right)^{1/\Gamma_2}, \quad k_\pm = \left( \frac{\Gamma_1 + \Gamma_2 \pm \sqrt{\Gamma_1^2 + \Gamma_1 \Gamma_2 + \Gamma_2^2}}{\Gamma_2} \right)^2 \label{k_pm}
\end{equation}
and $k_0$ is either $k_-$ or $k_+$ such that  
\begin{equation*}
 k_0 = \argmin_{k \in \{ k_+, k_-\}} \psi\left( \frac{\Gamma_1}{\Gamma_2} k\right). 
\end{equation*}
Then we have the following main theorem.

\begin{theorem}
Let $h \in C^1(\mathbb{R}^2)$ be a positive radial function satisfying (\ref{h-asympt}), (\ref{h^eps-1}), $\chi_{3+\eta} h \in L^\infty(\mathbb{R}^2)$ with $\eta > 0$, $\chi_1 \nabla h \in L^1(\mathbb{R}^2)$ and $h_r^\prime < 0$. Suppose (\ref{condi1}) and the constant $\mathscr{H}_c$ satisfies
\begin{equation}
\frac{\Gamma_1^2 \Gamma_2^2}{4\pi(\Gamma_1 + \Gamma_2)} \log\left( \psi\left( \frac{\Gamma_1}{\Gamma_2}k_0 \right) \left[ \psi\left( \frac{\Gamma_1}{\Gamma_2} \right) \right]^{-1} \right) < \mathscr{H}_c < 0.  \label{H-Range}
\end{equation}
We also assume that, for any initial configuration with $\mathscr{H}^\varepsilon = \mathscr{H}_c$, the corresponding solution of (\ref{L_123})  does not converge to a relative equilibrium as either of $t \rightarrow \pm \infty$. Then, there exists a constant $t^\ast$ such that $l^\varepsilon_{mn}(t^\ast) \rightarrow 0$ as  $\varepsilon \rightarrow 0$ and
\begin{equation*}
\lim_{\varepsilon \rightarrow 0} \mathscr{Z}^\varepsilon = - z_0 \delta(\cdot - t^\ast) , \qquad \lim_{\varepsilon \rightarrow 0} \mathscr{D}_E^\varepsilon = 0
\end{equation*}
in the sense of distributions, where  
\begin{equation}
z_0 = \int_{-\infty}^\infty \mathscr{Z}_0(\tau) d\tau. 
\label{z_0}
\end{equation} 
\label{thm-dissipation}
\end{theorem}

Theorem~\ref{thm-dissipation} asserts that the enstrophy variation converges to the $\delta$-measure with the mass of $- z_0$ as $\varepsilon \rightarrow 0$. In other words, the total enstrophy variation converges to
the Heaviside function $\mathcal{H}$ as follows.
\begin{equation}
\int_{-\infty}^t \mathscr{Z}^\varepsilon(\tau) d\tau \longrightarrow -z_0 \mathcal{H}(t-t^\ast).
\label{Heaviside}
\end{equation}
If $z_0>0$, the enstrophy dissipation occurs. Let us here note that the solution to the EP-PV system is time reversible, since it
is a Hamiltonian  system. Hence,  as discussed in \cite{G.2}, even if the direction of time is reversed, we have  the same convergence (\ref{Heaviside}),
which claims that the self-similar triple collapse always dissipates the enstrophy as long as $z_0>0$.
This is the emergence of the irreversibility of time direction in the conservative dynamical system.  However, it is still unknown 
whether or not the enstrophy always dissipates in that limit, since the sign of $z_0$ has not yet been determined. 
The following corollary  gives a sufficient condition for the enstrophy dissipation, which is described in terms of the function $Z(r)$ 
coming from (\ref{Z_nm}):
\begin{equation}
Z(r) = \int_0^\infty s \left\vert 2\pi\widehat{h}(s) \right\vert^2  J_0 \left( s \sqrt{r} \right) ds. \label{Z}
\end{equation}

\begin{corollary}
Suppose that $Z(r)$ is monotone decreasing and downward-convex. Then, for any initial configuration satisfying the assumptions of Theorem~\ref{thm-dissipation} and $M \geq 0$, we have $z_0 > 0$. For the case of $M < 0$, if the functions $Z(r)$ and $H_P(r)$ satisfy the additional condition
\begin{equation}
Z^{\prime\prime}(r) H^\prime_P(r) - Z^\prime(r) H^{\prime\prime}_P(r) > 0,  \label{condi-dissipation}
\end{equation}
then we have $z_0 > 0$. 

\label{Dissipation}
\end{corollary}

\begin{remark}
While the EP-PV system and the $\alpha$PV system have the same Hamiltonian structure, the difference between them consists in the functions describing the Hamiltonian, whose 
correspondence are listed in Table~\ref{Tab_Function}. However, the following theorems and lemmas can be proven in the same way as those for the $\alpha$PV system in \cite{G.2}, 
since we just need to use the common properties shared with those functions in Table~\ref{Tab_Function}. Accordingly, the proofs are accomplished by formally replacing $B_K$, $K_0$
$h_0$, $h_K$ with $P_K$, $H_G$, $H_P$, $L_P$, which are not shown in this paper to avoid redundancy. 
\label{rem-proof}
\end{remark}
The following two theorems correspond to Theorem~2 and Theorem~3 of \cite{G.2}. 
\begin{theorem}
Under the same assumptions of Theorem~\ref{thm-dissipation}, in the $\varepsilon \rightarrow 0$ limit, the solution of (\ref{xEP-PV}) with $N =3$ converges to the self-similar collapsing solution for $t < t^\ast$ and the expanding solution for $t > t^\ast$ with the same value of the Hamiltonian $\mathscr{H}_c$ in the three point-vortex system.
\label{convergence}
\end{theorem}
This indicates that the three $\varepsilon$-vortex points collapse self-similarly at $t=t^\ast$ in the $\varepsilon \rightarrow 0$ limit. Hence the enstrophy dissipation occurs at the event of the collapse. 
As a matter of fact, (\ref{condi1}) is the necessary condition for the existence of the enstrophy dissipation via the triple collapse, which is stated as follows.
\begin{theorem}
Suppose $L_{mn}(t) \rightarrow +\infty$ as $t \rightarrow \pm \infty$ for $m \neq n$. Then, (\ref{condi1}) holds.  \label{necessary}
\end{theorem}

\begin{remark}
For $M<0$, as we see in Section~\ref{proof-cor} and Section~\ref{VBM},  the following lemmas help us to check the condition (\ref{condi-dissipation}) whose proofs are the same as those of Lemma~5, Lemma~6 and Proposition~7 of \cite{G.2}. 
\begin{lemma}
Suppose (\ref{condi1}) and $M < 0$. Then, if $\mathscr{H} \leq 0$, then we have either $\gamma_1 > 1 > \gamma_2$ or $\gamma_2 > 1 > \gamma_1$, where $\gamma_1=L_{23}/L_{12}$ and
$\gamma_2 = L_{31}/L_{12}$, and these relations can not change throughout the evolution.
\label{RangeH-}
\end{lemma}

\begin{lemma}
Suppose (\ref{condi1}), $M < 0$ and $\mathscr{H} \leq 0$. Then, every level curve of the Hamiltonian starting from collinear configurations
 is monotone increasing as a function of $L_{23}^2$ and it  asymptotically approaches infinity along a straight line as $L_{23}^2 \rightarrow \infty$. 
\label{asympt-}
\end{lemma}
\end{remark}

\subsection{Proof of Theorem~\ref{thm-dissipation}}
\label{proof}

Formally, it is easy to show the convergence of $\mathscr{Z}^\varepsilon$ and $\mathscr{D}^\varepsilon$ in the sense of distributions. 
For any compactly supported smooth function $\phi(\tau)$, if $\mathscr{Z}_0$ and $\mbox{d}E_0/\mbox{d}\tau$ decay rapidly enough to be integrable on $\mathbb{R}$  
and $E_0$ vanishes at infinity,  we have
\begin{align*}
\langle \mathscr{Z}^\varepsilon, \phi \rangle & = - \int_{-\infty}^\infty \frac{1}{\varepsilon^2} \mathscr{Z}_0 \left( \frac{t - t^\ast}{\varepsilon^2} \right) \phi (t) dt = - \int_{-\infty}^\infty \mathscr{Z}_0(\tau) \phi (\varepsilon^2 \tau + t^\ast) d\tau \\
& \rightarrow - \phi(t^\ast) \int_{-\infty}^\infty \mathscr{Z}_0(\tau)d\tau = - z_0 \phi(t^\ast), 
\end{align*}
and
\begin{align*}
\langle \mathscr{D}_E^\varepsilon, \phi \rangle & = \int_{-\infty}^\infty \frac{1}{\varepsilon^2} \frac{dE_0}{d\tau} \left( \frac{t - t^\ast}{\varepsilon^2} \right) \phi (t) dt  = \int_{-\infty}^\infty \frac{dE_0}{d\tau} (\tau) \phi (\varepsilon^2 \tau + t^\ast) d\tau \\
& \rightarrow \phi(t^\ast) \int_{-\infty}^\infty \frac{dE_0}{d\tau}(\tau)d\tau = \phi(t^\ast) \left[ E_0(\tau) \right]_{-\infty}^\infty = 0,
\end{align*}
as $\varepsilon \rightarrow 0$. In order to make above argument mathematically rigorous, since both $\mathscr{Z}_0$ and $\mbox{d}E_0/\mbox{d}\tau$ are continuous functions on $\mathbb{R}$, we
will show that those functions decay rapidly, and $E_0$ vanishes at infinity.

We first remark that any solution of (\ref{XEP-PV}) subject to the assumptions of Theorem~\ref{thm-dissipation} satisfies
\begin{equation}
L_{mn}(t) \sim  \mathcal{O}(\vert t\vert^{1/2}), \qquad t \rightarrow \pm \infty.
\label{L-asympt}
\end{equation}
This asymptotic behavior is obtained by investigating the level curves of the Hamiltonian. Its proof proceeds in the same way as
that of Section~4 in \cite{G.2} as we mention in Remark~\ref{rem-proof}.

To show that the value of $z_0$ given in (\ref{z_0}) is well-defined, it is sufficient to see
 that $Z_{mn}(\tau)$ is integrable on $\mathbb{R}$. Considering the relation between the Fourier transform (\ref{FT}) and the Hankel transform (\ref{HT}), we find
\begin{equation*}
Z_{mn} = \int_0^\infty s \left\vert 2\pi\widehat{h}(s) \right\vert^2  J_0 \left( s L_{mn} \right) ds = (2 \pi)^2 \mathscr{F}\left[ \vert \widehat{h} \vert^2 \right]({\bm x}),
\end{equation*}
in which $L_{mn} = \vert{\bm x}\vert$. It follows from 
\begin{equation*}
\vert \widehat{f} \vert^2 = \left( \mathscr{F}\left[ f \right] \right)^2 = \frac{1}{2 \pi} \mathscr{F}\left[ f \ast f \right], \qquad \mathscr{F}[f] = \frac{1}{2 \pi} \mathscr{F}^{-1}[f]
\end{equation*}
for any radial function $f$ that we obtain
\begin{equation*}
Z_{mn}(\tau) = (2 \pi)^2 \mathscr{F}\left[ \vert \widehat{h}\vert^2 \right]({\bm x}) = \mathscr{F}^{-1}\left[ \mathscr{F}\left[ h \ast h \right] \right]({\bm x}) = \int_{\mathbb{R}^2} h({\bm x} - {\bm y}) h({\bm y}) d{\bm y}.
\end{equation*}
Since $h_r(r)$ is positive and monotone decreasing, it follows that
\begin{align*}
\int_{\mathbb{R}^2} h_r(\vert{\bm x} - {\bm y}\vert) h_r(\vert{\bm y}\vert) d{\bm y} &= \int_{\vert{\bm y}\vert \leq \vert{\bm x}\vert/2} h_r(\vert{\bm x} - {\bm y}\vert) h_r(\vert{\bm y}\vert) d{\bm y} + \int_{\vert{\bm y}\vert > \vert{\bm x}\vert/2} h_r(\vert{\bm x} - {\bm y}\vert) h_r(\vert{\bm y}\vert) d{\bm y} \\
&\leq h_r\left( \frac{\vert{\bm x}\vert}{2} \right)  \left[ \int_{\vert{\bm y}\vert \leq \vert{\bm x}\vert/2} h_r(\vert{\bm y}\vert) d{\bm y} + \int_{\vert{\bm y}\vert > \vert{\bm x}\vert/2} h_r(\vert{\bm x} - {\bm y}\vert)  d{\bm y} \right] \\
&\leq 2 \Vert h \Vert_{L^1} h_r\left( \frac{\vert{\bm x}\vert}{2} \right).
\end{align*}
Then, for sufficiently large $\tau_0 > 0$, owing to (\ref{L-asympt}),  we have
\begin{align*}
\int_{\tau_0}^\infty Z_{mn}(\tau)d\tau &= \int_{\tau_0}^\infty \int_{\mathbb{R}^2} h({\bm x} - {\bm y}) h({\bm y}) d{\bm y} d\tau \leq 2 \Vert h \Vert_{L^1} \int_{\tau_0}^\infty h_r\left( \frac{L_{mn}(\tau)}{2} \right) d\tau  \\
&\leq 2 \Vert h \Vert_{L^1} \int_{\tau_0}^\infty h_r\left( c_0 \tau^{1/2} \right) d\tau = c \Vert h \Vert_{L^1} \int_{c_0 \tau_0^{1/2}}^\infty r h_r(r) dr \leq c \Vert h \Vert_{L^1}^2 
\end{align*}
and similarly
\begin{equation*}
\int_{-\infty}^{-\tau_0}Z_{mn}(\tau)d\tau = \int_{- \infty}^{-\tau_0} \int_{\mathbb{R}^2} h({\bm x}- {\bm y}) h({\bm y}) d{\bm y} d\tau \leq c \Vert h \Vert_{L^1}^2.
\end{equation*}
It is easy to check that $h \in L^2(\mathbb{R}^2)$ and
\begin{equation*}
\int^{\tau_0}_{-\tau_0} \int_{\mathbb{R}^2} h({\bm x}- {\bm y}) h({\bm y}) d{\bm y} d\tau \leq c \Vert h \Vert_{L^2}^2.
\end{equation*}
Therefore, we conclude that $Z_{mn}(\tau)$ is integrable on $\mathbb{R}$ so that $z_0$ is finite.

Next, we show that $E_0(\tau)$ vanishes at infinity and $\mbox{d}E_0/\mbox{d}\tau$ is integrable on $\mathbb{R}$. Similar to the enstrophy, we confirm those claims for $E_{mn}$ instead of $E_0$. Let us recall the definition of $E_{mn}$:
\begin{align*}
E_{mn}(\tau) &= - H_G\left( L_{mn}(\tau) \right) +  \int_0^\infty \frac{1}{s} \left( 1 - \left\vert 2\pi\widehat{h}(s) \right\vert^2 \right) J_0 \left( s L_{mn}(\tau) \right) ds \\
& \equiv - H_G\left( L_{mn}(\tau) \right) + E_J\left( L_{mn}(\tau) \right),
\end{align*}
where
\begin{equation}
E_J(r) = \int_0^\infty \frac{1}{s} \left( 1 - \left\vert 2\pi\widehat{h}(s) \right\vert^2 \right) J_0 \left( s r \right) ds. 
\label{E_J}
\end{equation}
Regarding the function $H_G(r)$ in (\ref{H_G}), we have
\begin{equation*}
H_G(\vert{\bm x}\vert) = - \log{\vert{\bm x}\vert} - 2 \pi G_r^1(\vert{\bm x}\vert) = \int_{\mathbb{R}^2} \left( \log{\frac{\vert{\bm x} - {\bm y}\vert}{\vert{\bm x}\vert}} \right) h({\bm y}) d{\bm y}.
\end{equation*}
By dividing $\mathbb{R}^2$ into three domains and evaluating the integration on these domains separately, we obtain
\begin{align*}
\left\vert H_G(\vert{\bm x}\vert) \right\vert \leq& \int_{ \vert\vert{\bm y}\vert - \vert{\bm x}\vert\vert \leq \eta \vert{\bm x}\vert} \left\vert \log{\frac{\vert{\bm x} - {\bm y}\vert}{\vert{\bm x}\vert}} \right\vert h({\bm y}) d{\bm y} \\
& + \int_{ \eta \vert{\bm x}\vert < \vert\vert{\bm y}\vert - \vert{\bm x}\vert\vert \leq \vert{\bm x}\vert/ \eta} \left\vert \log{\frac{\vert{\bm x} - {\bm y}\vert}{\vert{\bm x}\vert}} \right\vert h({\bm y}) d{\bm y} \\
& + \int_{ \vert\vert{\bm y}\vert - \vert{\bm x}\vert\vert > \vert{\bm x}\vert/ \eta} \left\vert \log{\frac{\vert{\bm x} - {\bm y}\vert}{\vert{\bm x}\vert}} \right\vert h({\bm y}) d{\bm y}.
\end{align*}
The first term is estimated as
\begin{align*}
\int_{ \vert\vert{\bm y}\vert - \vert{\bm x}\vert\vert \leq \eta \vert{\bm x}\vert} \left\vert \log{\frac{\vert{\bm x} - {\bm y}\vert}{\vert{\bm x}\vert}} \right\vert h({\bm y}) d{\bm y} &= \int_{\left\vert \vert{\bm z}\vert - 1 \right\vert \leq \eta} \left\vert \log{ \left\vert\frac{{\bm x}}{\vert{\bm x}\vert} - {\bm z} \right\vert} \right\vert h(\vert{\bm x}\vert{\bm z}) \vert{\bm x}\vert^2 d{\bm z} \\
&\leq \frac{\Vert \chi_3 h \Vert_{L^\infty}}{\vert{\bm x}\vert} \int_{ \left\vert \vert{\bm z}\vert - 1 \right\vert \leq \eta} \frac{1}{\vert{\bm z}\vert^3}\left\vert \log{ \left\vert\frac{{\bm x}}{\vert{\bm x}\vert} - {\bm z} \right\vert} \right\vert d{\bm z} \\
&\leq \frac{\Vert \chi_3 h \Vert_{L^\infty}}{\vert{\bm x}\vert(1 - \eta)^3} \int_{ \left\vert \vert{\bm z}\vert - 1 \right\vert \leq \eta} \left\vert \log{ \left\vert\frac{{\bm x}}{\vert{\bm x}\vert} - {\bm z} \right\vert} \right\vert d{\bm z} \leq \frac{c_\eta}{\vert{\bm x}\vert} \Vert \chi_3 h \Vert_{L^\infty}.
\end{align*}
Since $\eta \vert{\bm x} \vert< \vert\vert {\bm y} \vert - \vert{\bm x}\vert\vert \leq \vert{\bm x}\vert/ \eta $ yields $\eta < \vert{\bm x} - {\bm y}\vert / \vert{\bm x}\vert < 2 + \eta^{-1}$, we estimate the second term as
\[
\int_{ \eta \vert{\bm x}\vert < \vert\vert{\bm y}\vert - \vert{\bm x}\vert\vert \leq \vert{\bm x}\vert/ \eta} \left\vert \log{\frac{\vert{\bm x} - {\bm y}\vert}{\vert{\bm x}\vert}} \right\vert h({\bm y}) d{\bm y} \leq c_\eta \int_{ \eta \vert{\bm x}\vert < \vert\vert{\bm y}\vert - \vert{\bm x}\vert\vert  \leq \vert{\bm x}\vert/ \eta} h({\bm y}) d{\bm y}  \leq \frac{c_\eta}{\vert{\bm x}\vert} \Vert \chi_1 h \Vert_{L^1}.
\]
Finally, the third term is estimated as
\begin{align*}
\int_{ \vert\vert{\bm y}\vert - \vert{\bm x}\vert\vert > \vert{\bm x}\vert/ \eta} \left\vert \log{\frac{\vert{\bm x} - {\bm y}\vert}{\vert{\bm x}\vert}} \right\vert h({\bm y}) d{\bm y} & = \int_{ \vert\vert{\bm z}\vert - 1\vert > 1 / \eta} \left\vert \log{ \left\vert\frac{{\bm x}}{\vert{\bm x}\vert} - {\bm z} \right\vert} \right\vert h( \vert{\bm x}\vert {\bm z}) \vert{\bm x}\vert^2 d{\bm z} \\
&\leq \int_{ \vert{\bm z}\vert > 1 + 1 / \eta}  \log{ \left( 1 + \vert{\bm z}\vert \right) } h( \vert{\bm x}\vert {\bm z}) \vert{\bm x}\vert^2 d{\bm z} \\
&\leq \int_{ \vert{\bm z}\vert > 1 + 1 / \eta} \vert{\bm z}\vert h( \vert{\bm x}\vert {\bm z}) \vert{\bm x}\vert^2 d{\bm z} \\
&\leq \frac{1}{\vert{\bm x}\vert} \int_{ \vert{\bm z}\vert > (1 + 1 / \eta)\vert{\bm x}\vert} \vert{\bm z}\vert h( {\bm z})  d{\bm z} \leq \frac{1}{\vert{\bm x}\vert} \Vert \chi_1 h \Vert_{L^1}.
\end{align*}
Remembering that $h$ satisfies $\chi_3 h \in L^\infty(\mathbb{R}^2)$ and $\chi_1 h \in L^1(\mathbb{R}^2)$ under the assumptions of Theorem~\ref{thm-dissipation}, we obtain
\begin{equation*}
\left\vert H_G(L_{mn}) \right\vert \leq \frac{c}{L_{mn}} \left( \Vert \chi_3 h \Vert_{L^\infty} + \Vert \chi_1 h \Vert_{L^1} \right).
\end{equation*}
In order to investigate the function $E_J(r)$, let us observe the properties of the Fourier and Hankel transforms. It is easy to check that $\Vert \mathscr{F}[f] \Vert_{L^\infty} \leq (2 \pi)^{-1} \Vert f \Vert_{L^1}$ and 
\begin{equation*}
\vert {\bm k} \vert \frac{\mbox{d} \widehat{f}}{\mbox{d}s}( \vert {\bm k}\vert) = - \frac{1}{2 \pi}\int_{\mathbb{R}^2} \left( 2 f(\vert{\bm x}\vert) + \vert{\bm x}\vert f^\prime(\vert{\bm x}\vert) \right) e^{- i {\bm x} \cdot {\bm k}} d{\bm x}.
\end{equation*}
Then, we find $\Vert\widehat{h}\Vert_{L^\infty} \leq (2 \pi)^{-1} \Vert h \Vert_{L^1}$ and $\Vert \chi_1 \mbox{d}\widehat{h}/\mbox{d}s \Vert_{L^\infty} \leq (2 \pi)^{-1} ( 2 \Vert h \Vert_{L^1} + \Vert \chi_1 \nabla h \Vert_{L^1})$, and we also have $\Vert \mbox{d}\widehat{h}/\mbox{d}s \Vert_{L^\infty} \leq (2 \pi)^{-1} \Vert \chi_1 h \Vert_{L^1}$. Since $\widehat{h}(0)=1$ owing to (\ref{h^eps-1}), the mean value theorem yields
\begin{equation*}
\left\vert 2 \pi \widehat{h}(s) - 1 \right\vert = \left\vert 2 \pi \widehat{h}(s) - 2\pi\widehat{h}(0) \right\vert \leq 2 \pi s \int_0^1 \left\vert \frac{\mbox{d}\widehat{h}}{\mbox{d}s}(\tau s) \right\vert d\tau.
\end{equation*}
Hence, we have the following estimate for (\ref{E_J}).
\begin{align*}
\left\vert E_J(\vert{\bm x}\vert) \right\vert &= \left\vert\int_0^\infty \frac{1}{s} \left( 1 - 2\pi\widehat{h}(s) \right)\left( 1 + 2\pi\widehat{h}(s) \right) J_0 \left( s \vert{\bm x}\vert \right) ds \right\vert \\
&\leq 2 \pi \int_0^\infty  \int_0^1 \left\vert \frac{\mbox{d}\widehat{h}}{\mbox{d}s}(\tau s) \right\vert d\tau \left( 1 + 2\pi\left\vert \widehat{h}(s) \right\vert \right) \left\vert J_0 \left( s \vert{\bm x}\vert \right) \right\vert ds \\
&=  \frac{2 \pi}{\vert{\bm x}\vert} \int_0^\infty  \int_0^1 \left\vert \frac{\mbox{d}\widehat{h}}{\mbox{d}s}\left( \frac{u}{\vert{\bm x}\vert}\tau \right) \right\vert d\tau \left( 1 + 2\pi\left\vert \widehat{h}\left( \frac{u}{\vert{\bm x}\vert} \right) \right\vert \right) \left\vert J_0 \left( u \right) \right\vert du.
\end{align*}
Setting the constant $1 /2 < \alpha < 1$, we obtain
\begin{align*}
\left\vert E_J(\vert{\bm x}\vert) \right\vert &\leq  \frac{2 \pi}{\vert{\bm x}\vert} \left\Vert \chi_{\alpha} \frac{\mbox{d}\widehat{h}}{\mbox{d}s} \right\Vert_{L^\infty} \int_0^1 \tau^{-\alpha}d\tau \int_0^\infty \frac{\vert{\bm x}\vert^\alpha}{u^\alpha}   \left( 1 + 2\pi\left\vert \widehat{h}\left( \frac{u}{\vert{\bm x}\vert} \right) \right\vert \right) \left\vert J_0 \left( u \right) \right\vert du \\
&\leq  \frac{c}{\vert{\bm x}\vert^{1 - \alpha}} \left\Vert \chi_{\alpha} \frac{\mbox{d}\widehat{h}}{\mbox{d}s} \right\Vert_{L^\infty} \left( 1 + \Vert \widehat{h} \Vert_{L^\infty} \right) \int_0^1 \tau^{-\alpha}d\tau \int_0^\infty \frac{1}{u^\alpha} \left\vert J_0 \left( u \right) \right\vert du,
\end{align*}
in which the rightmost integral with respect $u$ is well-defined owing to $J_0(0)=1$ and $J_0(r) \sim r^{-1/2}$ as $r \rightarrow \infty$. Note that it follows from the inequality $\vert{\bm x}\vert^\alpha \leq 1 + \vert{\bm x}\vert $ for $1/2 < \alpha < 1$ that
\begin{equation}
\left\Vert \chi_{\alpha} \frac{\mbox{d}\widehat{h}}{\mbox{d}s} \right\Vert_{L^\infty} \leq \left\Vert \frac{\mbox{d}\widehat{h}}{\mbox{d}s} \right\Vert_{L^\infty} + \left\Vert \chi_1 \frac{\mbox{d}\widehat{h}}{\mbox{d}s} \right\Vert_{L^\infty} \leq \frac{1}{2 \pi} \left( 2 \Vert h \Vert_{L^1} + \Vert \chi_1 h \Vert_{L^1} + \Vert \chi_1 \nabla h \Vert_{L^1} \right).
\end{equation}
Hence, we have
\begin{equation*}
\left\vert E_J(L_{mn}) \right\vert \leq  \frac{c}{L_{mn}^{1 - \alpha}} \left( \Vert h \Vert_{L^1} + \Vert \chi_1 h \Vert_{L^1} + \Vert \chi_1 \nabla h \Vert_{L^1} \right) \left( 1 + \Vert \widehat{h} \Vert_{L^\infty} \right).
\end{equation*}
Combining the estimates for $H_G(r)$ and $E_J(r)$ and considering (\ref{L-asympt}), we conclude that $E_{mn}(\tau)$ vanishes at infinity with the order $\mathcal{O}(\vert\tau\vert^{ - (1 -\alpha)/2})$, $1 /2 < \alpha < 1$ so that $E_0(\tau)$ decays with the same order.

We finally show that $\mbox{d}E_{mn}/\mbox{d}t$ is integrable on $\mathbb{R}$. Regarding the derivative of $H_G(r)$, it follows from $ \vert \mbox{d}H_G (\vert{\bm x} \vert)/\mbox{d}r \vert = \vert \nabla H_G(\vert {\bm x} \vert)\vert$ that 
\begin{align*}
\left\vert \frac{\mbox{d}H_G}{\mbox{d}r}(\vert{\bm x}\vert) \right\vert &\leq \int_{\mathbb{R}^2} \left\vert \frac{{\bm x} - {\bm y}}{\vert{\bm x} - {\bm y}\vert^2} - \frac{{\bm x}}{\vert{\bm x}\vert^2} \right\vert h({\bm y}) d{\bm y} = \frac{1}{\vert{\bm x}\vert} \int_{\mathbb{R}^2} \frac{\vert{\bm y}\vert}{\vert{\bm x} - {\bm y}\vert} h({\bm y}) d{\bm y} \\
&\leq \frac{1}{\vert{\bm x}\vert^2} \int_{\mathbb{R}^2} \left( 1 + \frac{\vert{\bm y}\vert}{\vert{\bm x} - {\bm y}\vert} \right) \vert{\bm y}\vert h({\bm y}) d{\bm y} \leq \frac{1}{\vert{\bm x}\vert^2} \left[ \Vert\chi_1 h \Vert_{L^1} + \int_{\mathbb{R}^2} \frac{\vert{\bm y}\vert^2}{\vert{\bm x} - {\bm y}\vert} h({\bm y}) d{\bm y} \right]. 
\end{align*}
The second term in the right-hand side is estimated as follows.
\begin{align*}
\int_{\mathbb{R}^2} \frac{\vert{\bm y}\vert^2}{\vert{\bm x} - {\bm y}\vert} h({\bm y}) d{\bm y} &\leq \Vert \chi_2 h \Vert_{L^\infty} \left[ \int_{\vert{\bm x} - {\bm y}\vert \leq 1} \frac{1}{\vert{\bm x} - {\bm y}\vert} d{\bm y} + \int_{\vert{\bm x} - {\bm y}\vert > 1 \, \cap \, \vert{\bm y}\vert \leq 1} \frac{1}{\vert{\bm x} - {\bm y}\vert} d{\bm y} \right] \\
&\quad  + \Vert \chi_{3 + \eta} h \Vert_{L^\infty} \left[ \int_{\vert{\bm x} - {\bm y}\vert >  \vert{\bm y}\vert > 1} \frac{1}{\vert{\bm y}\vert^{2 + \eta}}  d{\bm y} +  \int_{ \vert{\bm y}\vert > \vert{\bm x} - {\bm y}\vert > 1} \frac{1}{\vert{\bm x} - {\bm y}\vert^{2 + \eta}}  d{\bm y} \right] \\
&\leq c \left( \Vert \chi_2 h \Vert_{L^\infty} +  \Vert \chi_{3 + \eta} h \Vert_{L^\infty} \right).
\end{align*}
We  thus obtain
\begin{equation*}
\left\vert \frac{\mbox{d}H_G}{\mbox{d}r}(\vert{\bm x}\vert) \right\vert \leq \frac{c}{\vert{\bm x}\vert^2} \left( \Vert\chi_1 h \Vert_{L^1} + \Vert \chi_2 h \Vert_{L^\infty} +  \Vert \chi_{3 + \eta} h \Vert_{L^\infty} \right), 
\end{equation*}
and we know that right-hand side is finite owing to the assumptions of $h$. In order to estimate the derivative of $E_J(r)$, we rewrite (\ref{E_J}) by
\begin{align*}
E_J(\vert{\bm x}\vert) = \int_0^\infty \frac{1}{u} \left( 1 - \left\vert 2\pi\widehat{h}\left( \frac{u}{\vert{\bm x}\vert} \right) \right\vert^2 \right) J_0 \left( u \right) du.
\end{align*}
Its derivative is then expressed by
\begin{align*}
\frac{\mbox{d}E_J}{\mbox{d}r}(\vert{\bm x}\vert) = \frac{8 \pi^2}{\vert{\bm x}\vert^2} \int_0^\infty \widehat{h}\left( \frac{u}{\vert{\bm x}\vert} \right) \frac{\mbox{d}\widehat{h}}{\mbox{d}s}\left( \frac{u}{\vert{\bm x}\vert} \right) J_0 \left( u \right) du.
\end{align*}
Similar to the calculation estimating $E_J(r)$, we have
\begin{align*}
\left\vert \frac{\mbox{d}E_J}{\mbox{d}r}(\vert{\bm x}\vert) \right\vert &= \frac{8 \pi^2}{\vert{\bm x}\vert^{2-\alpha}} \int_0^\infty  \left\vert \widehat{h}\left( \frac{u}{\vert{\bm x}\vert} \right) \right\vert \left( \frac{u}{\vert{\bm x}\vert} \right)^\alpha \left\vert \frac{\mbox{d}\widehat{h}}{\mbox{d}s}\left( \frac{u}{\vert{\bm x}\vert} \right) \right\vert \frac{1}{u^\alpha} \left\vert J_0 \left( u \right) \right\vert du \\
&\leq \frac{c}{\vert{\bm x}\vert^{2 - \alpha}} \Vert \widehat{h} \Vert_{L^\infty} \left\Vert \chi_\alpha \frac{\mbox{d}\widehat{h}}{\mbox{d}s} \right\Vert_{L^\infty} \leq \frac{c}{\vert{\bm x}\vert^{2 - \alpha}} \left( \Vert h \Vert_{L^1} + \Vert \chi_1 h \Vert_{L^1} + \Vert \chi_1 \nabla h \Vert_{L^1} \right) \Vert \widehat{h} \Vert_{L^\infty}, 
\end{align*}
in which $1 /2 < \alpha < 1$. Summarizing the estimates and (\ref{L-asympt}), we find 
\begin{align*}
\left\vert \frac{\mbox{d} E_{mn}}{\mbox{d}\tau}(\tau) \right\vert &= \left\vert \left(- \frac{\mbox{d} H_G}{\mbox{d}r}(L_{mn}(\tau)) + \frac{\mbox{d} E_J}{\mbox{d}r}(L_{mn}(\tau)) \right)\frac{1}{2 L_{mn}(\tau)} \right\vert \left\vert \frac{\mbox{d} L_{mn}^2}{\mbox{d} \tau}(\tau)\right\vert \\
& \leq \frac{c}{(L_{mn}(\tau))^{3 - \alpha}} \sim \mathcal{O}(\vert\tau\vert^{- 1 - (1 - \alpha) /2 }), \qquad \tau \rightarrow \pm \infty.
\end{align*}
Consequently,  $\mbox{d} E_{mn} / \mbox{d}\tau$ and $\mbox{d} E_0 / \mbox{d}\tau$ are integrable on $\mathbb{R}$.

\subsection{Proof of Corollary~\ref{Dissipation}}
\label{proof-cor}
In order to show $z_0 > 0$, it is enough to prove that $\mathscr{Z}_0$ is a positive function. Note that, owing to (\ref{condi1}), $\mathscr{Z}_0$ is rewritten by
\begin{equation*}
\mathscr{Z}_0 = \frac{\Gamma_1 \Gamma_2}{2 \pi} \left( \frac{\Gamma_2}{\Gamma_1 + \Gamma_2} Z( L_{23}^2) + \frac{\Gamma_1}{\Gamma_1 + \Gamma_2} Z(L_{31}^2) - Z(L_{12}^2) \right)
\end{equation*}
and (\ref{M}) is equivalent to 
\begin{equation*}
\frac{\Gamma_2}{\Gamma_1 + \Gamma_2} L_{23}^2 + \frac{\Gamma_1}{\Gamma_1 + \Gamma_2} L_{31} = L_{12}^2 - \frac{M}{\Gamma_1 \Gamma_2}.
\end{equation*}
Since $Z(r)$ is monotone decreasing and downward-convex, it follows from $M \geq 0$ that 
\begin{equation*}
\mathscr{Z}_0 > \frac{\Gamma_1 \Gamma_2}{2 \pi} \left[ Z\left( L_{12}^2 - \frac{M}{\Gamma_1 \Gamma_2} \right) - Z(L_{12}^2) \right] \geq 0.
\end{equation*}
Hence, we  easily obtain the result for $M \geq 0$ as desired. To see the case of $M < 0$, let us introduce the following notations,
\begin{equation*}
L = L_{12}^2, \qquad \mu = \frac{L_{23}^2}{L_{12}^2}, \qquad \widehat{\mu} = \frac{L_{31}^2}{L_{12}^2}.
\end{equation*}
According to Lemma~\ref{RangeH-}, under the conditions $M < 0$ and $\mathscr{H} < 0$, either $\mu < 1 < \widehat{\mu}$ or $\widehat{\mu} < 1 < \mu$ holds true for all time.
 We here consider the case $\mu < 1 < \widehat{\mu}$. As for the other case, we can show the fact similarly. Note that  (\ref{M}) implies
\begin{equation*}
\widehat{\mu} = \widehat{\mu}(\mu) = - \frac{\Gamma_2}{\Gamma_1} \mu + \frac{1}{\Gamma_3 \Gamma_1 L} \left( M - \Gamma_1 \Gamma_2 L \right).
\end{equation*}
For any fixed $L > 0$ and $M < 0$, we rewrite $\mathscr{Z}_0$ as the function of $\mu$.
\begin{equation*}
\mathscr{Z}_0 = Z_{L, M}(\mu) = \frac{\Gamma_1 \Gamma_2}{2 \pi} \left( \frac{\Gamma_2}{\Gamma_1 + \Gamma_2} Z( \mu L) + \frac{\Gamma_1}{\Gamma_1 + \Gamma_2} Z( \widehat{\mu}(\mu) L) - Z(L) \right).
\end{equation*}
Then, we have
\begin{equation*}
\frac{\mbox{d} Z_{L,M}}{\mbox{d} \mu} (\mu) = \frac{\Gamma_1 \Gamma_2^2 L}{2 \pi (\Gamma_1 + \Gamma_2)} \left( Z^\prime( \mu L) - Z^\prime(\widehat{\mu}(\mu) L)  \right) < 0,
\end{equation*}
since $Z^\prime(r)$ is negative and monotone increasing. Similarly, since the Hamiltonian (\ref{H}) is expressed by 
\begin{equation*}
\mathscr{H} = H_{L,M}(\mu) = \frac{\Gamma_1 \Gamma_2}{2 \pi} \left( \frac{\Gamma_2}{\Gamma_1 + \Gamma_2} H_P( \mu L) + \frac{\Gamma_1}{\Gamma_1 + \Gamma_2} H_P( \widehat{\mu}(\mu) L) - H_P(L) \right),
\end{equation*}
we have
\begin{equation*}
\frac{\mbox{d} H_{L,M}}{\mbox{d} \mu} (\mu) = \frac{\Gamma_1 \Gamma_2^2 L}{4 \pi (\Gamma_1 + \Gamma_2)} \left( L_P( \mu L) - L_P( \widehat{\mu}(\mu) L)  \right) > 0
\end{equation*}
 owing to (\ref{dL_p}). Hence, it is sufficient to show that $Z_{L,M}(\mu_0) > 0$ for the constant $\mu_0$ satisfying $H_{L,M}(\mu_0) = 0$. Indeed, since $\mathscr{H} = H_{L,M}(\mu) < 0$ is equivalent to $\mu < \mu_0$, we obtain $\mathscr{Z}_0 = Z_{L, M}(\mu) > Z_{L, M}(\mu_0) > 0$ for $\mu < \mu_0$, i.e. $\mathscr{Z}_0$ is always positive when $\mathscr{H}$ is negative.

Since $H_P(r)$ is continuous and monotone increasing, as shown in Appendix~\ref{functions}.4, there exists the inverse function $H_P^{-1}(r)$ so that the function  
$Z_H(r) = Z(H_P^{-1}(r))$ is well-defined. Then, the first and second derivatives of $Z_H(r)$ are expressed by
\begin{align*}
\frac{\mbox{d}Z_H}{\mbox{d}r}(r) &=  \frac{Z^\prime(H_P^{-1}(r))}{H_P^\prime(H_P^{-1}(r))}, \\
\frac{\mbox{d}^2 Z_H}{\mbox{d}r^2}(r) &= \frac{Z^{\prime\prime}(H_P^{-1}(r))}{\left( H_P^{\prime}(H_P^{-1}(r)) \right)^2} - Z^\prime(H_P^{-1}(r)) \frac{ H^{\prime\prime}_P(H_P^{-1}(r))}{\left( H_P^\prime(H_P^{-1}(r)) \right)^3} = \frac{DZ(H_P^{-1}(r))}{\left( H_P^\prime(H_P^{-1}(r)) \right)^3}, 
\end{align*}
where $DZ(r) = Z^{\prime\prime}(r) H_P^\prime(r) - Z^\prime(r) H^{\prime\prime}_P(r)$. Note that $H^\prime_P(r) = L_P(r) / 2 $ is a positive function. Owing to (\ref{condi-dissipation}) and positivity of $H^\prime_P(r)$, we find $\mbox{d}^2Z_H / \mbox{d}r^2 > 0$ so that $Z_H(r)$ is downward-convex. Hence, considering $H_{L,M}(\mu_0) = 0$, namely  
\begin{equation*}
 \frac{\Gamma_2}{\Gamma_1 + \Gamma_2} H_P( \mu_0 L) + \frac{\Gamma_1}{\Gamma_1 + \Gamma_2} H_P( \widehat{\mu}(\mu_0) L) = H_P(L),
\end{equation*}
we obtain
\begin{align*}
Z_{L, M}(\mu_0) &= \frac{\Gamma_1 \Gamma_2}{2 \pi} \left[ \frac{\Gamma_2}{\Gamma_1 + \Gamma_2} Z_H\left( H_P(\mu_0 L) \right) + \frac{\Gamma_1}{\Gamma_1 + \Gamma_2} Z_H\left( H_P(\widehat{\mu}(\mu_0) L) \right) - Z_H\left( H_P(L) \right) \right] \\
&> \frac{\Gamma_1 \Gamma_2}{2 \pi} \left[ Z_H\left( \frac{\Gamma_2}{\Gamma_1 + \Gamma_2} H_P( \mu_0 L) + \frac{\Gamma_1}{\Gamma_1 + \Gamma_2} H_P( \widehat{\mu}(\mu_0) L) ) \right) - Z_H\left( H_P(L) \right) \right] \\
&= \frac{\Gamma_1 \Gamma_2}{2 \pi} \left[ Z_H\left( H_P(L) ) \right) - Z_H\left( H_P(L) \right) \right] = 0.
\end{align*}
Since we assume $\mathscr{H} <0$, we achieve the conclusion.

\section{Applications to various smoothing functions}
\label{VBM}
We apply the main results to some smoothing functions for the Euler-Poinvar\'{e} models to show the existence of the anomalous enstrophy dissipation via the triple collapse.
\subsection{Gaussian kernel}
The simplest smoothing function for the Euler-Poincare model is the Gaussian kernel, which is given by
\begin{equation}
h({\bm x}) = \frac{1}{\pi} e^{- \vert {\bm x} \vert^2}. \label{gauss}
\end{equation}
It is easy to confirm that $h$ satisfies the assumptions of Theorem~\ref{thm-dissipation}. In order to prove $z_0 > 0$, we show that the functions $Z$ and $H_p$, which are derived from $h$ as (\ref{Z}) and (\ref{H_P}) respectively, satisfy the assumptions of Corollary~\ref{Dissipation}. The function $Z(r)$ is monotone decreasing and downward convex, since we have
\[
Z(r) = \int_{\mathbb{R}^2} h({\bm x} - {\bm y}) h({\bm y}) d{\bm y} =\frac{1}{\pi^2}  e^{ - \frac{1}{2}\vert {\bm x} \vert^2} \int_{\mathbb{R}^2} e^{ - 2 \vert {\bm y} - \frac{1}{2}{\bm x} \vert^2} d{\bm y}=\frac{1}{2 \pi}  e^{ - \frac{1}{2}\vert {\bm x} \vert^2} = \frac{1}{2 \pi}  e^{ - \frac{1}{2} r},
\]
in which $r = \vert {\bm x} \vert^2$. Regarding the sufficient condition (\ref{condi-dissipation}), since 
\begin{equation*}
Z^{\prime\prime}(r) H^\prime_P(r) - Z^\prime(r) H^{\prime\prime}_P(r) = \frac{1}{8 \pi} \mbox{e}^{-\frac{1}{2}r}\left( H_P^\prime(r) + 2 H_P^{\prime\prime}(r) \right) \equiv \frac{1}{8 \pi} \mbox{e}^{-\frac{1}{2}r} DH_P(r),
\end{equation*}
it is enough to show that $DH_P(r)$ is a positive function. Owing to (\ref{Apdx_P_K}) and (\ref{Apdx_H_P}), $DH_P(r)$ is expressed by
\[
DH_P(r) = \frac{1}{2 r^2} \left( (r -2)P_K(\sqrt{r}) + 2 r e^{- r} \right) \equiv \frac{1}{2 r^2} p_0(r).
\]
It follows from $p_0^\prime(r) = P_K(\sqrt{r}) - r e^{-r}$ and $p_0^{\prime\prime}(r) = r e^{-r} > 0$ with $p_0^{\prime}(0) = 0$ that 
$p_0^{\prime}$ is positive  and so $p_0$ is monotone increasing. In addition, owing to $p_0(0) = 0$, we find $p_0 > 0$ and thus $DH_P > 0$. Therefore, we conclude that the enstrophy dissipates in the EP-PV system with (\ref{gauss}).

\subsection{The vortex blob system}
We consider the vortex blob regularization as another important example of the Euler-Poincar\'{e} models, in which the smoothing function $h^\sigma$ is given by
\begin{equation*}
h^\sigma({\bm x}) = \frac{1}{\sigma^2} h_b \left( \frac{{\bm x}}{\sigma} \right), \qquad h_b({\bm x}) = \frac{1}{\pi (\vert{\bm x}\vert^2 + 1)^2}.
\end{equation*} 
Since $h_b$ satisfies the assumptions of Theorem~\ref{thm-dissipation}, the enstrophy variation   in the three point vortex problem of the $\sigma$PV system converges to the $\delta$-measure with  the mass of $-z_0$ in the $\sigma \rightarrow 0$ limit. In what follows, we confirm numerically that the constant $z_0$ is strictly positive. For $r = \vert{\bm x}\vert^2$, the function $Z(r)$ is given by
\begin{align*}
Z(r)  &= \int_{\mathbb{R}^2} h_b({\bm x} - {\bm y}) h_b({\bm y}) d{\bm y} = \frac{1}{\pi^2}\int_{\mathbb{R}^2} \frac{1}{\left( \vert{\bm x} - {\bm y}\vert^2 + 1 \right)^2} \frac{1}{\left( \vert{\bm y}\vert^2 + 1 \right)^2} d{\bm y} \\
&= \frac{1}{\pi^2} \int_0^\infty \frac{s}{(s^2 + 1)^2} \int_0^{2\pi} \frac{1}{\left( r + s^2 + 1 - 2 r^{1/2} s \cos{\theta} \right)^2}  d\theta ds \\
&= \frac{2}{\pi} \int_0^\infty \frac{s}{(s^2 + 1)^2} \frac{r + s^2 + 1}{\left( \left( r + s^2 + 1 \right)^2 - 4 r s^2 \right)^{3/2}} ds = \frac{1}{\pi} \int_1^\infty \frac{1}{s^2} \frac{r + s}{\left( \left( r - s \right)^2 + 4 r \right)^{3/2}} ds. 
\end{align*}
As shown in Figure~\ref{Z_delta}, since 
\begin{align*}
Z^\prime(r) &= - \frac{2}{\pi} \int_1^\infty \frac{1}{s^2} \frac{r^2 + (s + 1)r - 2 s^2 + 3s}{\left( \left( r - s \right)^2 + 4 r \right)^{5/2}} ds < 0, \\
Z^{\prime\prime}(r) &= \frac{6}{\pi} \int_1^\infty \frac{1}{s^2} \frac{ r^3 + (s + 2)r^2 + (-5s^2 +6s +2)r + 3s^3 -12s^2 + 10s}{\left( \left( r - s \right)^2 + 4 r \right)^{7/2}} ds >0,
\end{align*}
$Z(r)$ is monotone decreasing and downward-convex.  Hence, we conclude that $z_0$ is positive for $M \geq 0$.

\begin{figure}[tbhp]
\begin{tabular}{c}
\begin{minipage}{0.45\hsize}
\includegraphics[scale=0.5]{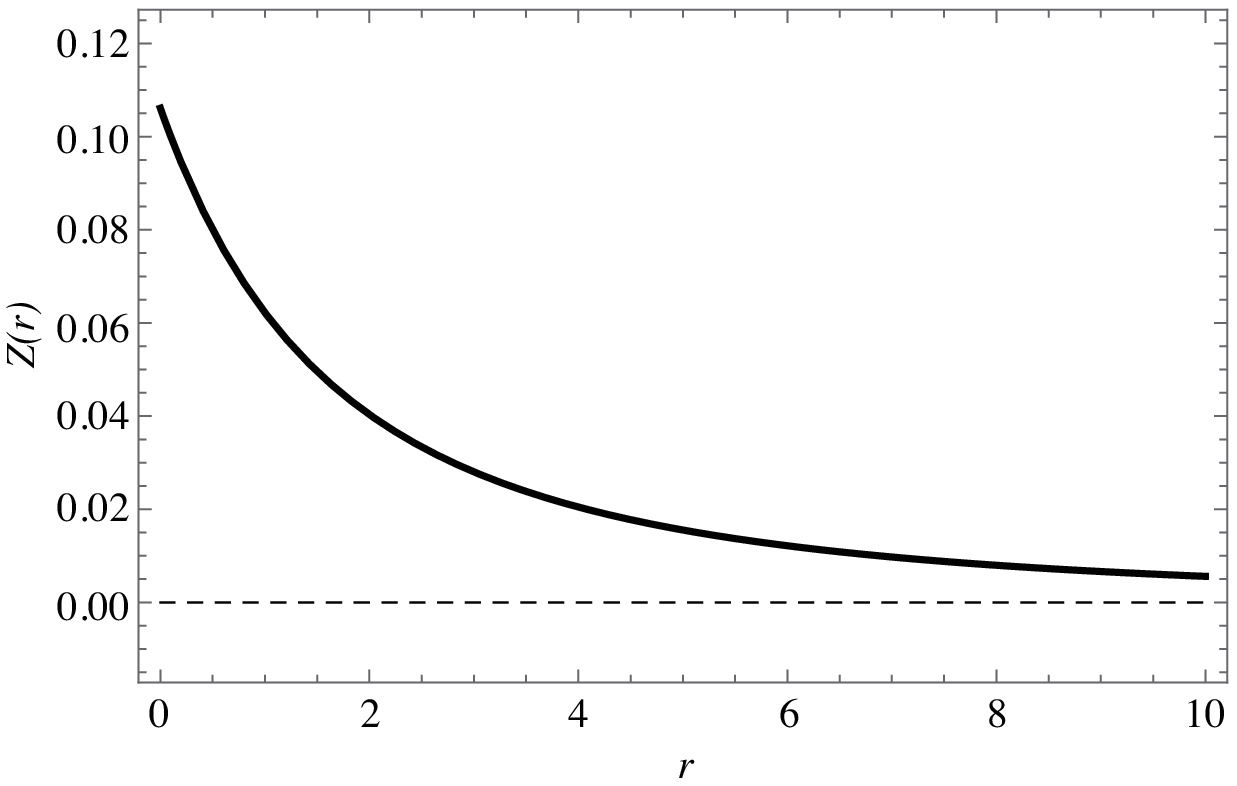}
\end{minipage}
\begin{minipage}{0.45\hsize}
\includegraphics[scale=0.5]{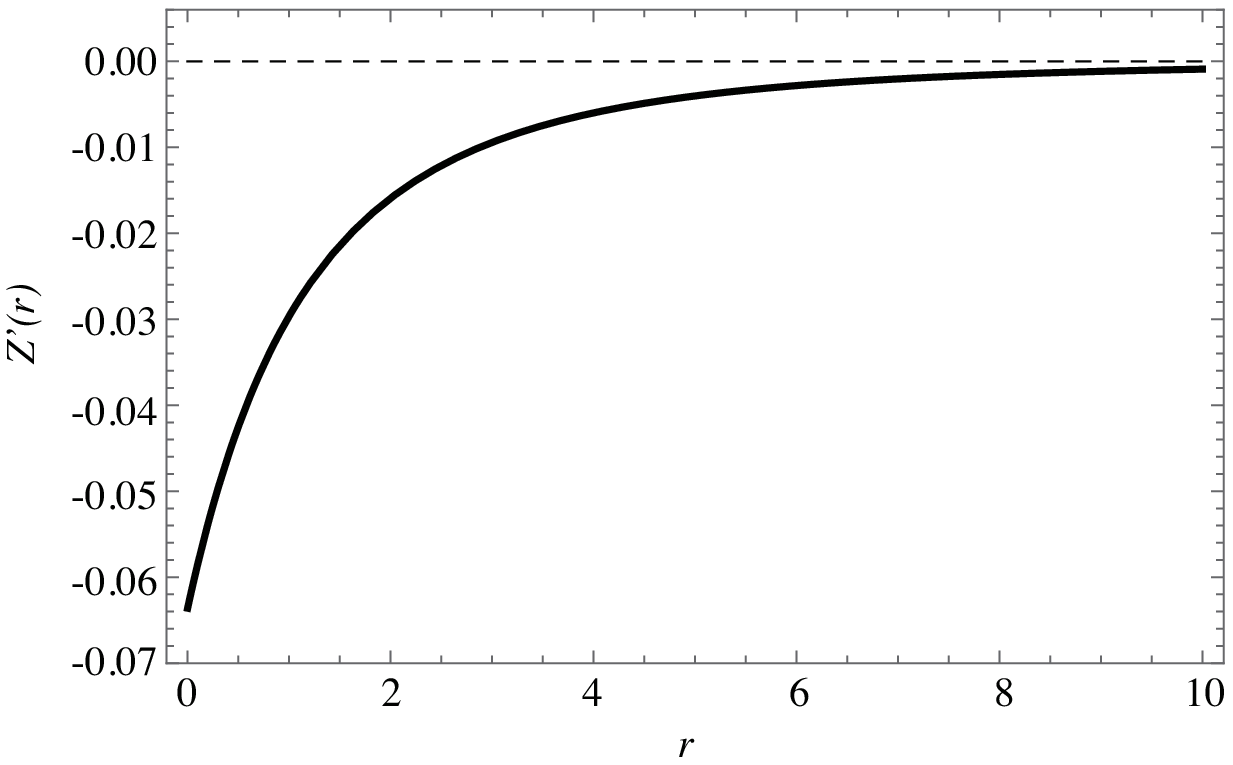}
\end{minipage}
 \\ 
\begin{minipage}{0.45\hsize}
\includegraphics[scale=0.5]{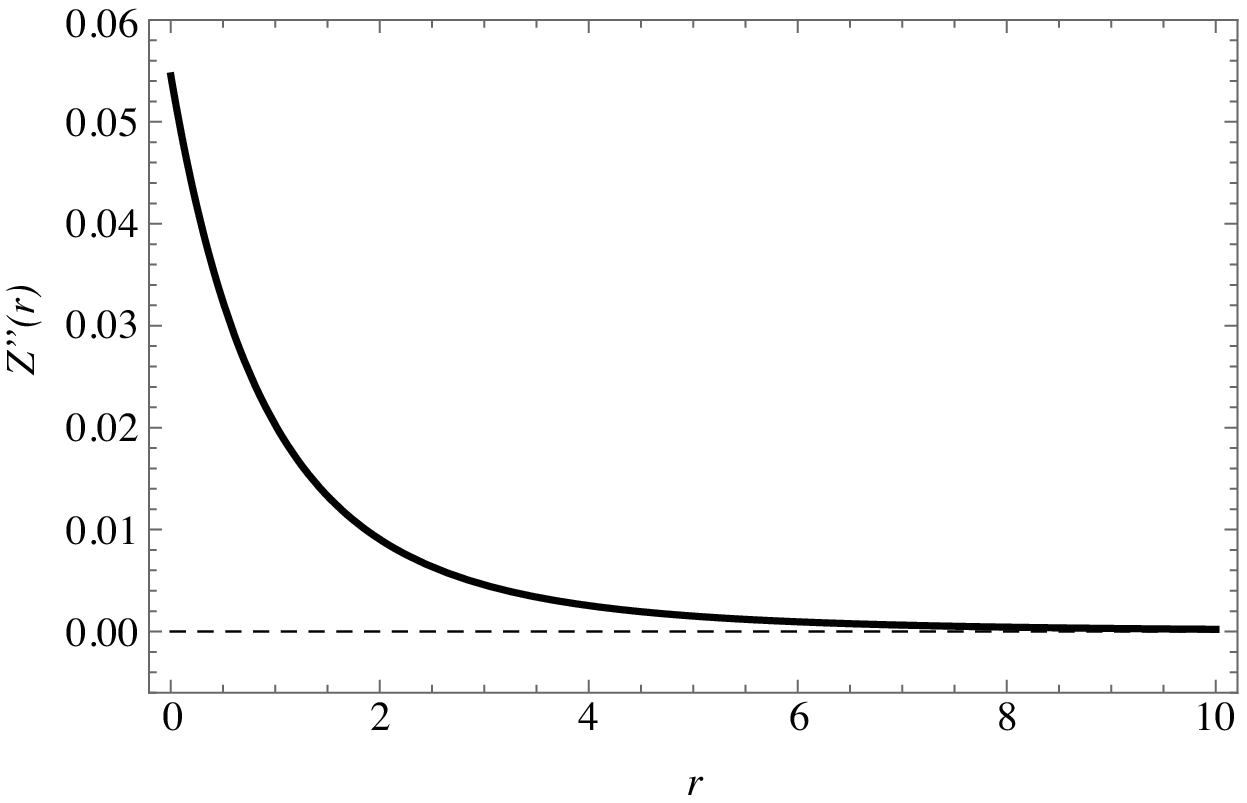}
\end{minipage}
\begin{minipage}{0.45\hsize}
\includegraphics[scale=0.5]{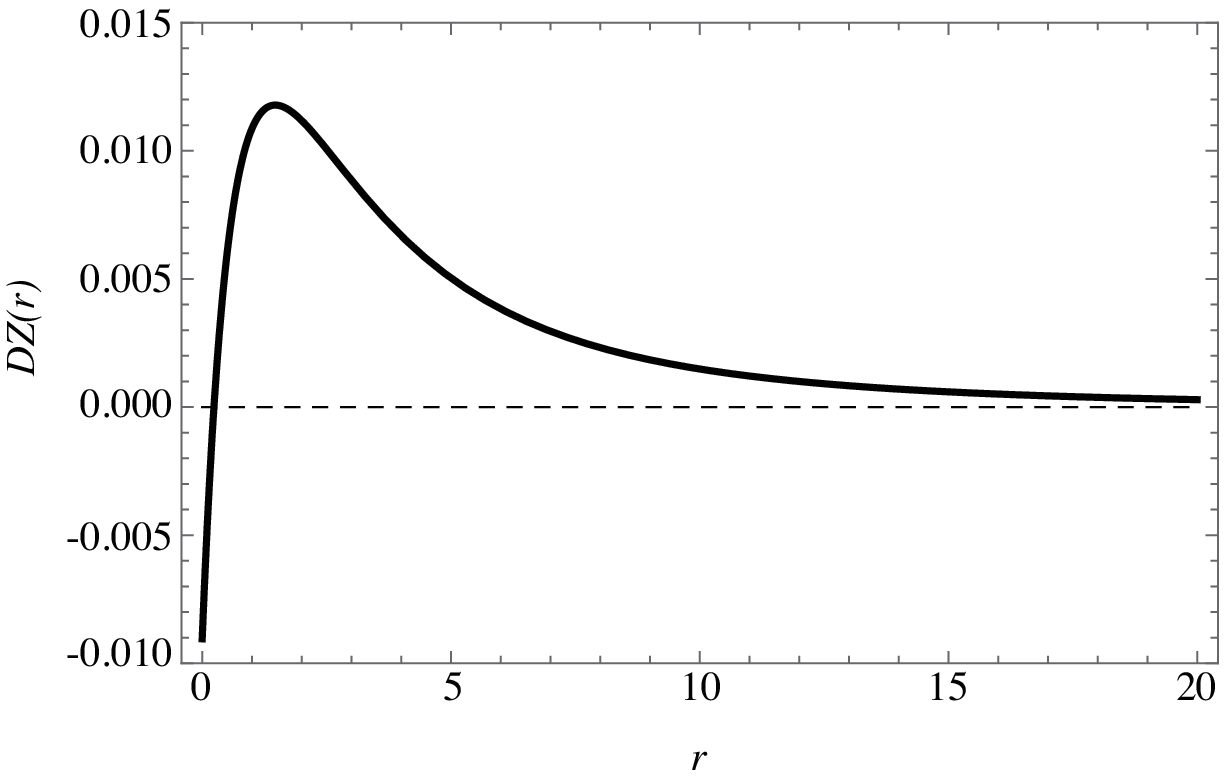}
\end{minipage}
\end{tabular}
\caption{Plots of the functions $Z(r)$, $Z^\prime(r)$, $Z^{\prime\prime}(r)$ and $DZ(r)$ for the vortex-blob regularization.}
\label{Z_delta}
\end{figure}

Next, we see the case of $M < 0$. In the vortex blob method, the Hamiltonian is expressed by
\begin{equation*}
\mathscr{H}^\sigma = - \frac{1}{2 \pi} \sum_{n=1}^N \sum_{m=n+1}^N \Gamma_n \Gamma_m \log{\sqrt{\left(l_{mn}^\sigma\right)^2 + 1}},  
\end{equation*}
and thus $H_P(r) = \log{\sqrt{r + 1}}$. The condition (\ref{condi-dissipation}) is then equivalent to
\begin{equation*}
DZ(r) \equiv Z^{\prime\prime}(r) (r + 1) + Z^\prime(r) > 0.
\end{equation*}
Figure~\ref{Z_delta} also shows the graph of $DZ(r)$, where there exists $r_0>0$ such that $DZ(r_0) = 0$ and $DZ(r) > 0$ for $r > r_0$. On the other hand, according to Lemma~\ref{asympt-}, when three $\sigma$-point vortices with any $M < 0$ and $\mathscr{H} < 0$ starts from any collinear configuration at the initial moment,  the distance $L_{mn}$ achieves its minimal value. That is to say, if we consider the initial data satisfying $L^2_{mn}(0) > r_0$ for any $m \neq n$, then $DZ(L_{mn})$ is always positive throughout the time evolution. Thus, owing to Corollary~\ref{Dissipation}, we obtain $z_0 > 0$.

\section{Concluding Remarks}
\label{concluding}

We have introduced the EP-PV system describing the evolution of $\varepsilon$-point vortices in the Euler-Poincar\'{e} models and proven the existence of the evolution of the three $\varepsilon$-point vortices whose enstrophy varies at the triple collapse in the sense of distributions in the $\varepsilon \rightarrow 0$ limit. Moreover, we give a sufficient condition for the existence of the anomalous enstrophy dissipation via the triple collapse. All conditions are described in terms of the radial smoothing function $h({\bm x})=h_r(\vert{\bm x}\vert) \in C^1(\mathbb{R}^2) \cap W_1^1(\mathbb{R}^2)$; There exists a self-similar collapsing orbit of three $\varepsilon$-point vortices with the distributional enstrophy variation in the limit of $\varepsilon \rightarrow 0$, if $h$ is monotone decreasing ($h_r^\prime <0$) and it satisfies the logarithmic singularity condition in the neighborhood of the origin $\chi_{\log}^{+} h \in L^\infty(\mathbb{R}^2)$ and the decay rate conditions at infinity $\chi_1 \nabla h \in L^\infty(\mathbb{R}^2)$ and $\chi_{3+\eta} h \in L^\infty(\mathbb{R}^2)$ with $\eta >0$. In addition, the sufficient condition for the enstrophy variation being dissipative is described in terms of  $Z(r)$ and $H_P(r)$ that are defined from the smoothing function $h$. Those conditions are  applicable to many smoothing functions including the Euler-$\alpha$ model, the Gaussian model and the vortex-blob model as confirmed in \cite{G.2} and Section~\ref{VBM}. Hence, we conclude that the anomalous enstrophy dissipation via the collapse of three point vortices is universally constructed within the framework of the Euler-Poincar\'{e} models.

Let us finally mention the future direction. It is interesting to investigate the enstrophy variation together with the evolution of many $\varepsilon$-point vortices. According to \cite{Kimura}, the $N$ point vortices in the PV system can collapse self-similarly in finite time under certain circumstances. Thus, there is a possibility of obtaining the enstrophy dissipation by considering the collapse of the $N$ vortex problem in the EP-PV system. As a matter of fact, for the $\alpha$PV system, the enstrophy dissipation has been  observed numerically in \cite{G.1} via a quadruple self-similar collapse as $\alpha \rightarrow 0$. However, since the EP-PV system is not integrable for $N \geq 4$ in general, it is not an easy task to prove this. Further mathematical analysis is required.

\appendix
\section{Properties of auxiliary functions}
\label{functions}

We introduce some functions associated with a given smoothing function $h_r(r)$ that is positive and monotone decreasing.
Here, we show the properties of those functions that are essentially used in the proofs of the main results in the same way as
in \cite{G.2}. See also in Table~\ref{Tab_Function}.

\paragraph{1. The function $P_K(r)$} The function $P_K(r)$ defined by (\ref{P_K})
is monotone increasing and upward-convex. Note that the derivative of $P_K(r)$ is expressed by 
\begin{equation*}
\frac{\mbox{d}}{\mbox{d}r} P_K(r) = 2 \pi \frac{\mbox{d}}{\mbox{d}r} \left( - r \frac{\mbox{d}G_r^1}{\mbox{d}r} (r) \right) = 2 \pi r h_r(r),
\end{equation*}
since $G^1_r$ is a radial function and satisfies   $- \Delta G_r^1(\vert{\bm x}\vert) = h_r(\vert{\bm x}\vert)$. Hence, it follows that
\begin{equation}
\frac{\mbox{d}}{\mbox{d}r} P_K(\sqrt{r}) = \frac{1}{2 \sqrt{r}} \frac{\mbox{d} P_K}{\mbox{d}r}(\sqrt{r}) = \pi h_r(\sqrt{r}) > 0, \qquad 
\frac{\mbox{d}^2}{\mbox{d}r^2} P_K(\sqrt{r}) = \frac{\pi}{2 \sqrt{r}} h^\prime_r(\sqrt{r}) < 0.
\label{Apdx_P_K}
\end{equation}
Moreover, as we see in \cite{G.3}, $P_K(r)$ satisfies
\begin{equation*}
P_K(0) = 0, \qquad \lim_{r \rightarrow \infty} P_K(r) = \int_{\mathbb{R}^2} h({\bm x}) d{\bm x} =  1,
\end{equation*}
and we thus have $0 \leq P_K(r) < 1$.

\paragraph{2. The function $L_P(r)$} The function $L_P(r)$ defined by (\ref{L_P}) 
is  monotone decreasing. Indeed, its derivative is given by 
\begin{align*}
\frac{\mbox{d}}{\mbox{d}r} L_P(r) = - \frac{1}{r^2} P_K(\sqrt{r}) + \frac{\pi}{r} h_r(\sqrt{r}) \equiv - \frac{1}{r^2} l_0 (r), 
\end{align*}
where $l_0 (r) = P_K(\sqrt{r}) - \pi r h_r(\sqrt{r})$. Since it follows that
\begin{equation*}
\frac{\mbox{d}}{\mbox{d}r} l_0 (r) = - \frac{\pi}{2} \sqrt{r} h^\prime_r(\sqrt{r}) > 0, \qquad l_0 (0) = P_K(0) = 0, 
\end{equation*}
we find that $l_0(r)$ is a positive function and thus $L_P(r)$ is monotone decreasing.

\paragraph{3. The function $H_G(r)$} The function $H_G(r)$ defined by (\ref{H_G})
 is monotone decreasing and downward-convex. Owing to $0 \leq P_K(r) < 1$, the first and second derivatives of $H_G(\sqrt{r})$ are given by
\begin{equation*}
\frac{\mbox{d}}{\mbox{d}r} H_G(\sqrt{r})  = - \frac{1}{2 \sqrt{r}} \left( \frac{1}{\sqrt{r}} + 2 \pi \frac{dG_r^1}{dr}(\sqrt{r}) \right) = - \frac{1}{2 r} \left( 1 - P_K(\sqrt{r}) \right) < 0
\end{equation*}
and 
\begin{align*}
\frac{\mbox{d}^2}{\mbox{d}r^2} H_G(\sqrt{r}) & = \frac{1}{2 r^2} \left( 1 - P_K(\sqrt{r}) \right) + \frac{\pi}{2 r} h_r(\sqrt{r}) > 0.
\end{align*}
Note that
\begin{equation}
H_G(r) \sim - \log{r} - 2 \pi  G_r^1(0), \quad r \rightarrow 0, \label{H_G-zero}
\end{equation}
since $G_r^1(0)$ is finite.

\paragraph{4. The function $H_P(r)$} Its definition is given by (\ref{H_P}). It is monotone increasing and upward-convex, since we have
\begin{equation}
\frac{\mbox{d}}{\mbox{d}r} H_P(r) = \frac{1}{2 r} P_K(\sqrt{r}) = \frac{1}{2} L_P(r) > 0, \qquad \frac{\mbox{d}^2}{\mbox{d}r^2} H_P(r) = \frac{1}{2}\frac{\mbox{d}}{\mbox{d}r} L_P(r) < 0.
\label{Apdx_H_P}
\end{equation}

\section*{Acknowledgements}
This work is partially supported by JSPS A3 Foresight Program and Grants-in-Aid for Scientific Research KAKENHI (B) No. 26287023 from JSPS.


\end{document}